\documentclass[11pt,reqno]{amsart}

\usepackage{amsmath, amsthm, amscd, amsfonts, amssymb, graphicx, color, mathrsfs, extarrows}
\usepackage{mathscinet}
\usepackage[a4paper,left=3cm,right=3cm,top=2cm,bottom=4cm,bindingoffset=5mm]{geometry}

\usepackage{float}

\usepackage{eurosym}
\usepackage[utf8]{inputenc}
\usepackage{amsmath}
\usepackage{amsfonts}
\usepackage{version}
\usepackage{mathtools}
\usepackage{amscd}
\usepackage{amssymb}
\usepackage{dsfont}
\usepackage{amsthm}
\usepackage{thmtools}
\usepackage{graphicx}
\usepackage{mathrsfs}
\usepackage{todonotes}
\usepackage{xcolor}
\usepackage{todonotes}
\usepackage{setspace}
\usepackage{enumerate}
\usepackage[english]{babel}
\usepackage{xcolor}
\usepackage[colorlinks=true,linkcolor=blue]{hyperref}
\usepackage[all]{hypcap}

\newtheorem{theorem}{Theorem}
\newtheorem{proposition}[theorem]{Proposition}
\newtheorem{corollary}[theorem]{Corollary}
\newtheorem{lemma}[theorem]{Lemma}

\theoremstyle{definition}

\newtheorem{remark}[theorem]{Remark}
\newtheorem{example}[theorem]{Example}
\newtheorem{definition}[theorem]{Definition}

\newcommand{\cP}{\mathcal{P}}

\newcommand{\sA}{\mathscr{A}}

\newcommand{\Om}{\Omega}

\newcommand{\la}{\lambda}

\newcommand{\R}{\mathbb{R}}      
\newcommand{\N}{\mathbb{N}}
\renewcommand{\P}{{\mathbb P}}
\newcommand{\Q}{{\mathbb Q}}

\newcommand{\E}{\mathbb{E}}
\newcommand{\FF}{\mathcal F}

\renewcommand{\d}{{\rm d}}
\newcommand{\DR}{\varrho}
\newcommand{\PD}{{\rm PD}}
\newcommand{\moc}{{\rm MoC}}
\newcommand{\Bb}{{\rm B}_{\rm b}}
\newcommand{\Lb}{{\rm L}_{\rm b}}

\renewcommand{\d}{{\rm d}}
\newcommand{\eins}{\mathds{1}}

\DeclareMathOperator{\VaR}{\rm VaR}
\DeclareMathOperator{\ES}{\rm ES}
\DeclareMathOperator{\Exp}{\rm Exp}

\newcommand{\VR}{\VaR_\DR^\alpha}

\newcommand\JS[1]{{#1}}

\title[Default Risk Measures]{An axiomatic approach to default risk and model uncertainty in rating systems}

\author{Max Nendel}
\address{Center for Mathematical Economics, Bielefeld University, Germany}
\email{max.nendel@uni-bielefeld.de}

\author{Jan Streicher}
\address{Center for Mathematical Economics, Bielefeld University, Germany and Landesbank Baden-W\"urttemberg, Stuttgart, Germany}
\email{jan.streicher@uni-bielefeld.de}

\thanks{The authors thank Felix-Benedikt Liebrich, Markus Klein, Ruodu Wang, as well as two anonymous referees for helpful comments and discussions related to this work.\ This work was funded by the Deutsche Forschungsgemeinschaft (DFG, German Research Foundation) -- SFB 1283/2 2021 -- 317210226.\ The second author gratefully acknowledges the support of the Landesbank Baden-W\"urttemberg related to this work.}

\date{\today}

\begin{document}
\maketitle

\begin{abstract}
In this paper, we deal with an axiomatic approach to default risk.\ We introduce the notion of a default risk measure, which generalizes the classical probability of default (PD), and allows to incorporate model risk in various forms.\ We discuss different properties and representations of default risk measures via monetary risk measures, families of related tail risk measures, and Choquet capacities.\ In a second step, we turn our focus on default risk measures, which are given as worst-case PDs and distorted PDs.\ The latter are frequently used in order to take into account model risk for the computation of capital requirements through risk-weighted assets (RWAs), as demanded by the Capital Requirement Regulation (CRR). 
In this context, we discuss the impact of different default risk measures and margins of conservatism on the amount of risk-weighted assets.



\end{abstract}


\section{Introduction}
\JS{Financial institutions and corporate firms have a variety of different risk types, such as credit risk, market price risk, operational risk, or liquidity risk.\ These are ubiquitous in business operations and affect different business areas.\ In the context of corporate risk management, credit risk is of fundamental importance for a correct assessment of existing and potential future risks and an adequate management of business operations, in general.\ For most financial institutions, such as banks, credit risk is even the most important type of risk and therefore also in the focus of supervisory regulations such as the Capital Requirements Regulation (CRR) \cite{CRR}, which is part of Basel III and also known as Regulation (EU) 575/2013.}

\JS{From the institutes' point of view, credit risk generally involves the risk that borrowers cannot repay loans granted to them in accordance with the predetermined contracts. In the context of counterparty credit risk, ratings serve the purpose of dividing borrowers or business partners, in general, into risk classes, which are specified via a probability of default over a one-year horizon. Since default risk represents a complex issue, there are different rating systems depending on the type of customers (individuals, corporate customers, banks, insurance companies, etc.).\  For banks, ratings represent the basis of credit risk controlling and reporting, and lead to more differentiated pricing, greater competitiveness, and an improvement in credit decisions.} For a general introduction to credit risk modeling, we refer to Bluhm et al.\ \cite{bluhm2016introduction} and Lando\ \cite{lando2009credit}. \JS{We also refer to Guo et al.\ \cite{wang2021rating} for an axiomatic study of credit rating criteria.}

 \JS{The focus of a rating is on the probability of default (PD), which refers to the creditworthiness of the borrower, and \textit{not} to credit-specific terms, such as exposure at default (EaD) or loss given default (LGD). Nevertheless, all three terms play an important role in the context of default risk, are included in the calculation of expected losses (EL) and risk-weighted assets (RWAs), and are therefore part of supervisory requirements for risk-differentiated capital backing. Since rating systems are mathematical statistical models that transform a borrower's default-relevant characteristics into a statement of creditworthiness, they are subject to model risks and model uncertainties, which can lead to major discrepancies in credit risk management if neglected.\ Additionally, the European Banking Authority's (EBA) guidelines on PD and LGD estimation (EBA-GL-2017-16) \cite{EBA_GL} serve to reduce fluctuations in risk parameters, and focus on modeling techniques used in the estimation of risk parameters.}\ In particular, the PD estimation in low-default portfolios is a subtle issue, cf.\ Pluto and Tasche \cite{pluto2011estimating} and Tasche \cite{tasche2013bayesian}.
 
 In the aftermath of the subprime mortgage crisis, the topic of model uncertainty or Knightian uncertainty has become increasingly significant for financial institutions and found its way into regulatory requirements in various forms.\ \JS{As a consequence, this classical and already very prominent topic in economic theory} has received \JS{even more} attention in the literature on theoretical economics, mathematical finance, and actuarial sciences. \JS{Model uncertainty appears in the economic literature, for example,} in the context of preference relations, cf.\ Gilboa and Schmeidler \cite{MR1000102} and Maccheroni et al.\ \cite{MR2268407}, general equilibrium theory, cf.\ Beissner and Riedel \cite{MR3914797}, insurance pricing, cf.\ Castagnoli et al.\ \cite{MR2070175}, Nendel et al.\ \cite{MR4269594}, and Wang et al.\ \cite{MR1604936}, as well as hedging and no-arbitrage conditions, see, for instance, Bouchard and Nutz \cite{MR3313756} and Burzoni et al.\ \cite{MR4325181}. However, to the best of our knowledge, a detailed study of model uncertainty in credit risk management is not present in the literature.
 
 This paper therefore aims to provide a decision-theoretic foundation for the treatment of default risk and model uncertainty in rating systems.\ \JS{While the work of Guo et al.\ \cite{wang2021rating} focuses on an axiomatic study of credit rating criteria, including, among others, generalized PD criteria, our axiomatic approach introduces the notion of a default risk measure, which aims to provide a more general perspective on PDs, and allows to include, for example, model uncertainty in the form of worst-case PDs and distorted PDs, warning signals, and default risk arising from regulatory risk measures. In particular, we explore the use of generalized versions of PDs in the context of credit risk, but do not aim to identify relevant criteria for credit ratings.} In contrast to monetary risk measures, cf.\ Artzner et al.\ \cite{MR1850791}, Frittelli and Rosazza Gianin \cite{MR2766847}, and F\"ollmer and Schied \cite{MR3859905}, and nonlinear expectations, cf.\ Coquet et al. \cite{MR1906435} and Peng \cite{MR3970247}, default risk measures do not behave linearly along constants but only take the values zero (no default) or one (default) for constant functions. 
 
 Throughout, we consider a set $C$ of customers, i.e., a set of bounded measurable functions on a given measurable space $(\Omega,\mathcal F)$ containing all constant functions.\ A default risk measure $\DR$ is a monotone functional that assigns to each customer $X\in C$ a default risk $\DR(X)\in [0,1]$.\ Here, positive values of $X$ represent a negative total cash flow or, loosely speaking, a default.\ In a first step, we show that every default risk measure can be extended from the set $C$ to the space $\Bb$ of all bounded measurable functions, cf.\ Theorem \ref{thm.extension}.\ The extension procedure is constructive, and shows how default risk can be assigned consistently to new customers based on a financial institution's stock of existing clients.
 
 In a second step, we consider tail risk measures related to default risk measures.\ As noted by Liu and Wang \cite{MR4312589}, the consideration of tail risk, i.e., the risk beyond a given threshold, is crucial in today's financial regulation.\ We also refer to Bignozzi et al. \cite{riskmeasures} for a generalization of the value at risk that depends on the size of potential losses in the form of quantile-based risk measures, to Fadina et al.\ \cite{fadina2021one} for an axiomatic study of quantiles, and to Burzoni et al.\ \cite{BURZONI2022106297} for a study of adjusted expected shortfall. In Section \ref{sec.acceptance}, we show that each default risk measure induces its own notion of a value at risk, and establish a ono-to-one relation between default risk measures and so-called generalized quantile functions.\ For particular choices of default risk measures, e.g., distorted PDs and worst-case PDs, we provide explicit representations of the related value at risk. \JS{We point out that our notion of a generalized quantile function follows a different philosophy than the notion of a tail risk measure introduced by Liu and Wang \cite{MR4312589}.}

A key property of the PD is that it is specified only by the states of the world where a negative total cash flow is realized, independent of the amount of capital given liquidity or illiquidity.\ In mathematical terms, this means that the PD of $X$ is the same as the PD of $\eins_{\{X>0\}}$. In Section \ref{sec.properties}, we characterize default risk measures that have this property, and connect them to Choquet capacities, cf.\ Dellacherie and Meyer \cite{MR521810}. In this context, the notions of \JS{default} scaling invariance, liquidity invariance, and illiquidity invariance play a fundamental role. Using continuity properties of Choquet integrals, we derive sufficient and necessary conditions for default risk measures to admit a representation via probability measures or, in other words, as worst-case PDs.

For the calculation of RWAs, the concept of a margin of conservatism (MoC) is used\JS{, in practice,} to quantify the amount of model uncertainty.\ The regulatory need to consider model uncertainty regarding default risks by calculating a MoC that reflects the expected range of estimation errors can be found in Article 179 (f) or, PD-specific, in Article 180 (e) of the CRR \cite{CRR}, among others.\ In Section \ref{sec:distorted}, we characterize default risk measures that are given in terms of a MoC or, equivalently, as distorted PDs. The characterization generalizes the fact that law-invariant capacities on an atomless probability space can be represented as distorted probabilities, cf.\ Wang et al.\ \cite{MR1604936}, where this result is established in the context of insurance premia that are given as Choquet integrals and Amarante and Liebrich \cite{amarante2023distortion} for a detailed study of distortion risk measures, i.e., law-invariant and comonotonically additive risk measures.

Moreover, we establish a connection between distorted PDs and worst-case PDs or, equivalently, the margin of conservatism and a suitable set of probability measures, based on the Kusuoka representation of law-invariant risk measures, cf.\ Kusuoka \cite{MR1886557}, and the well-known Fr\'echet-Hoeffding bounds for joint distributions, cf.\ Burgert and R\"uschendorf \cite{MR2323193}, which also play a fundamental role for the collapse to the mean of law-invariant risk functionals, cf.\ Bellini et al.\ \cite{MR4236430} and Liebrich and Munari \cite{MR4446849}. 

\JS{In Section \ref{sec:application}, we use the results on distorted PDs in a case study on capital requirements as demanded by current regulations. There, we discuss the impact of model uncertainty in rating systems on financial institutions' RWAs. Since, from a regulatory perspective, model uncertainty only has to be considered for unexpected losses and not for expected losses (EL), ironically, a high degree of model uncertainty can actually reduce the amount of capital requirement for badly rated customers, since it transforms unexpected losses into expected ones. We refer to Example \ref{last ex} for the details.}

The rest of the paper is organized as follows.\ In Section \ref{sec: Definition}, we define the notion of a default risk measure, illustrate the definition in several examples, and state our extension result for default risk measures (Theorem \ref{thm.extension}).\ The link between default risk measures and generalized quantile functions is discussed in Section \ref{sec.acceptance}.\ Section \ref{sec.properties} is devoted to default risk measures that are given only in terms of default scenarios, see Theorem \ref{thm.main}.\ Section \ref{sc:Choquet} contains several results on robust representations as worst-case PDs. In Section \ref{sec:distorted}, we focus on law-invariant risk measures and distorted PDs. There, we connect distortion functions with certain properties to sets of absolutely continuous probability measures based on the value at risk and expected shortfall of probability densities. In Section \ref{sec:application}, we discuss capital requirements for rating systems, and illustrate the impact of different default risk measures on the amount of financial institutions' risk weighted assets. In Appendix \ref{app.distorted}, we provide a short proof for a characterization of exact capacities and distorted probabilities, cf.\ Aouani and Chateauneuf \cite{MR2442200} and Kadane and Wassermann \cite{MR1401848}. \JS{The proofs of Section \ref{sec: Definition} are contained in Appendix \ref{proofs of sc.2}. The proofs of Section \ref{sec.acceptance} can be found in Appendix \ref{proofs of sc.3}. The proofs of Section \ref{sec.properties} are collected in Appendix \ref{proofs of sc.4}. The proofs of Section \ref{sc:Choquet} are given in Appendix \ref{proofs of sc.5} and the proofs of Section \ref{sec:distorted} in Appendix \ref{proofs of sc.6}.}

\section{Default Risk Measures: Definition and Examples} \label{sec: Definition}

In this section, we introduce the concept of a default risk measure, which is strongly motivated by the probability of default (PD) as a prime example.\ Like monetary risk measures, default risk measures are monotone functionals defined on suitable sets of measurable functions, cf.\ \cite{MR3859905}.\ However, they exhibit a completely different behaviour along constants.

Throughout, let $(\Omega, \FF)$ be a measurable space and $\Bb=\Bb(\Omega,\FF)$ denote the space of all bounded measurable functions $\Omega\to \R$.\ We consider a set $C\subset \Bb$, containing the set of all constant functions.\ A function $X\in C$ can be interpreted as a customer of a financial institution with $-X(\omega)$ being the sum of all financial flows (earnings, spendings, and maturities combined) at the end of the respective observation period if a scenario $\omega\in\Omega$ is realized. Thus, positive values of $X$ resemble a negative sum of all financial flows, which we will, loosely speaking, refer to as a \textit{default}.\ Choosing this, in comparison to the literature on monetary risk measures, inverted sign convention leads to an easier exposition since it avoids confusion arising from repeated sign changes on several occasions.\ As in the theory of monetary risk measures, we do not differentiate between a real constant $m\in \R$ and the constant function $X\colon \Omega\to \R$ with $X(\omega)=m$ for all $\omega\in \Omega$, and write $X=m$, thinking of it as \textit{cash}. For $X,Y\in \Bb$, we write $X\leq Y$ if $X(\omega)\leq Y(\omega)$ for all $\omega\in \Omega$. Moreover, we define 
$$
\inf X:=\inf_{\omega\in \Omega}X(\omega)\quad \text{and}\quad\sup X:=\sup_{\omega\in \Omega}X(\omega)\quad\text{for all }X\in \Bb.
$$
Additionally, for $X\in \Bb$, we consider the standard decomposition $X=X^+-X^-$ with
$$X^+:=X\eins_{\{X>0\}}\quad \text{and}\quad X^-:=-X\eins_{\{X< 0\}}.$$
\JS{For any two real numbers $x,y\in \R$, we use the notation $x\vee y:=\max\{x,y\}$ and $x\wedge y:=\min\{x,y\}$.\ In a similar fashion, we write $X\vee Y$ and $X\wedge Y$ for the pointwise maximum and minimum of  $X,Y\in \Bb$, respectively.} 
Throughout, we use the following slightly modified notion of a (monetary) risk measure, and refer to \cite{MR3859905} for a detailed discussion on this topic.

\begin{definition}\label{def.riskmeasure}
 We say that a map $R\colon \Bb\to \R$ is a \textit{(monetary) risk measure} if
 \begin{enumerate}
    \item[(i)] $R(X)\leq R(Y)$ for all $X,Y\in \Bb$ with $X\leq Y$,
    \item[(ii)] $R(0)= 0$ and $R(X+m)=R(X)+m$ for all $X\in \Bb$ and $m\in \R$.
\end{enumerate}
\end{definition}

We now introduce the central object of our study. 

\begin{definition}\label{Def DRM}
A map $\DR \colon C \to [0,1]$ is called a \textit{default risk measure} if
\begin{enumerate}
    \item[(i)] $\DR(X)\leq \DR(Y)$ for all $X,Y\in C$ with $X\leq Y$,
    \item [(ii)] $\DR(0)= 0$ and $\DR(m)=1$ for all $m\in \R$ with $m>0$.
\end{enumerate}
\end{definition}

Thinking of PDs, the respective properties seem to be very canonical.\ For instance, comparing two customers it is obvious that the one with the higher total cash flow in all scenarios exhibits a lower risk of default (Property (i)). Moreover, for all $X\in C$ with $X\leq0$,
\[
0\leq \DR(X)\leq \DR(0)=0,
\]
i.e., if all obligations can be payed in any scenario the customer's default risk will be zero, and a constant negative total cash flow ($m>0$) leads at least to an unlikely repayment, and hence to a sure default (Property (ii)).\ Although Property (i) in the definition of a default risk measure is analogous to the monotonicity of monetary risk measures, Property (ii) is substantially different from the standard \textit{cash additivity} or \textit{translation invariance}.\ To that end, consider a default risk measure $\DR\colon \Bb\to [0,1]$ and observe that
\[
\DR(X-\sup X)=0\quad \text{for all }X\in \Bb.
\]
By definition, neither convexity nor positive homogeneity (of degree 1) are meaningful properties for default risk measures, since
$$\DR (\lambda)=1 > \lambda=\lambda \DR (1)\quad\text{for all }\lambda\in (0,1).$$
\JS{We thus observe that the properties of monetary risk measures differ substantially from those of default risk measures despite the similarity of their very general definitions. Nevertheless, there is the possibility to construct default risk measures from monetary risk measures as the following example illustrates.}
\begin{example}[Default risk measure defined by a monetary risk measure] 
Given a monetary risk measure $R\colon \Bb\to \R$ as in Definition \ref{def.riskmeasure}, we are able to construct a default risk measure via
$$\DR(X):= R\big(\eins_{\{X>0\}}\big)\quad\text{for all }X\in \Bb.$$ 
For $X,Y\in \Bb$ with $X\leq Y$, we have $\eins_{\{X>0\}}\leq \eins_{\{Y>0\}}$, so that 
$$0\leq R\big(\eins_{\{X>0\}}\big)\leq  R\big(\eins_{\{Y>0\}}\big)\leq 1.$$ 
Moreover,
$\DR(0)= R\big(\eins_{\emptyset}\big)=R(0)=0$ 
and $\DR(m)=R(1)=1$ for all $m\in \R$ with $m>0$.
\end{example}
\JS{We continue with several examples of default risk measures, and begin with the most prominent one}. 

\begin{example}[Probability of default]\label{ex:PD}  
We fix a reference probability measure $\P$ on $\FF$, and consider the \textit{probability of default} (PD), given by $$\PD_\P(X):=\P(X>0)\quad\text{for all }X\in \Bb.$$If the probability measure $\P$ is discrete, this default risk measure could be interpreted as a mapping to common rating classes.
Clearly, $\PD_\P$ satisfies all properties of a default risk measure. For $X,Y\in \Bb$ with $X \leq Y$, $\{X>0\}\subset\{Y>0\}$, so that $$0\leq \P(X>0)\leq \P(Y>0)\leq 1.$$ Furthermore,  $\PD_\P(0)=\P(\emptyset)=0$ and $\PD_\P(m)=\P(\Omega)=1$ for all $m\in \R$ with $m>0$.
\end{example}
\JS{Building on this example, we can also consider the case, where model uncertainty is taken into account via a distortion function.}
\begin{example}[Distorted PD]\label{Distorted PD} 
Again, we fix a reference probability measure $\P$ on $\FF$. Due to a lack of data, bad data quality, or changing economic environments, the consideration of uncertainties in form of a \textit{margin of conservatism} (MoC) becomes more and more important for financial institutions.\ Since such model uncertainties are part of any model, including rating models, it is possible that the reference probability measure $\P$ is not the `precise' probability measure that represents the default risk of costumers over a one-year time horizon.\
We therefore consider a nondecreasing distortion function $T\colon [0,1]\to [0,1]$ with $T(0)=0$ and $T(1)=1$. We define $$\DR(X):=T\big(\P(X>0)\big)=T\big(\PD_\P(X)\big)\quad\text{for all }X\in \Bb.$$
In this case, the distortion function $T$ can be regarded as a benchmark for model uncertainty, and the margin of conservatism is given by
$$\moc(p):=\frac{T(p)}{p}-1\quad\text{for all }p\in (0,1].$$  
Clearly, the two properties of a default risk measure carry over from the classical PD, cf.\ Example \ref{ex:PD}, to the distorted PD for any probability measure $\P$ and any nondecreasing distortion function $T\colon [0,1]\to [0,1]$ with $T(0)=0$ and $T(1)=1$.
\end{example}
\JS{Apart from distorting a reference probability measure as in the previous example, there is also the possibility of incorporating model uncertainty via worst-case considerations among sets of probability measures.}
\begin{example} [Worst-case PD]\label{WC PD}
 In order to properly account for uncertainties w.r.t.\ model specifications, it is often necessary to consider various models at the same time. This becomes particularly relevant, if the models have different sets of measure zero, since then one model neglects certain events that occur with positive probability under a different model.\ We therefore consider the following generalization of Example \ref{ex:PD}. Let $\cP$ be a nonempty set of probability measures and $$\DR(X):=\sup_{\Q\in\cP}\Q(X>0)=\sup_{\Q\in\cP}\PD_\Q(X)\quad\text{for all }X\in \Bb.$$ As before, the properties (i) and (ii) have been shown for classical PDs in Example \ref{ex:PD}, and remain valid when taking the supremum over PDs.\ Following \cite{MR2442200}, we will see that, in many cases, distorted PDs allow for a representation as worst-case PDs and vice versa, see Section \ref{sec:distorted}.
\end{example}

Up to now, all examples for default risk measures have been of the form
\begin{equation}\label{eq.repdrm}
\DR(X)=\DR\big(\eins_{\{X>0\}}\big)\quad \text{for all }X\in \Bb,
\end{equation}
i.e., $\DR(X)$ only depends on the set where $X\in \Bb$ is larger than zero, completely independent of its values.\ Thinking of PDs from rating systems, this is a very desirable property, since it implies that the financial institution is only interested in the customers' ability to pay their dues.\
In Section \ref{sec.properties}, we derive sufficient and necessary conditions for default risk measures in order to satisfy \eqref{eq.repdrm}.\ The following two examples show that, however, not every default risk measure needs to allow for such a representation \JS{as it is also possible to define default risk measures that depend on specific values of $X$, for instance, using a warning signal that leads to a more conservative risk assessment.}
\begin{example}[Warning signal]\label{WS}
Consider two default risk measures $\rho_0$ and $\rho_c$ with $\rho_0 \leq  \rho_c $, i.e., $\rho_0$ is less conservative then $\rho_c$. The idea is that $\rho_c$ acts as a warning signal if scenarios are possible where the loss exceeds a given maximum level. For $X\in \Bb$ and $\gamma>0$, we define 
 $$ \DR_{\gamma} (X):= \begin{cases}
  \DR_0(X), & \DR_c\big(X-\tfrac1\gamma\big)=0, \\
  \DR_c(X), &  \DR_c\big(X-\tfrac1\gamma\big)>0.
\end{cases} $$
Thus, for $\gamma>0$, we change from $ \DR_0$ to the more conservative default risk measure $\DR_c$ when the potential loss exceeds the level $\frac1\gamma$ under $\DR_c$. We observe that
$$ \lim_{\gamma\to 0}\DR_{\gamma}(X)=\DR_0(X)\quad\text{for all }X\in \Bb.$$ 
In fact, $\DR_\gamma(X)=0=\DR_0(X)$ for all $X\in \Bb$ with $X\leq 0$, and $\DR_c(X-\sup X)=0$ for all $X\in \Bb$ with $\sup X>0$.\ A concrete choice for $\DR_c$ and $\DR_0$ are, for example, $\DR_0=\PD_\P$ and  $\DR_c(x)=\sup_{\Q\in\cP}\PD_\Q(X)$ (worst-case PD) for all $X\in \Bb$, where $\cP$ is a nonempty set of probability measures containing $\P$. For example, it is conceivable that customers, for whom the possible loss exceeds the limit $\frac1\gamma$, might exhibit additional risk factors that increase their probability of default.
\end{example}
\JS{At first glance, the following example is reminiscent of Example \ref{WC PD}, in which the worst-case PD was considered. As an additional criterion, the risk of $X\in \Bb$, described by a monetary risk measure, determines how many models are considered in the calculation of the supremum.}
\begin{example}[Increasing conservatism]\label{ex.increasing}
Let $\mathcal P$ be a set of probability measures on $\FF$, $R\colon \Bb\to \R$ be a monetary risk measure, and $\alpha\colon \mathcal P\to [0,\infty)$ with $\inf_{\Q\in \mathcal P}\alpha(\Q)=0$. For $X\in \Bb$, let $\DR(X):=0$ if $R(X)\leq 0$ and
\[
\DR(X):=\sup\big\{\PD_\Q(X)\,\big|\, \Q\in \mathcal P,\, \alpha(\Q)\leq R(X)\big\} \quad\text{if }R(X)>0,
\]
i.e., the larger the risk associated to $X\in \Bb$, the more models are taken into account, when assessing the default risk.\ In this case, $\alpha(\Q)$ measures the degree of confidence that the model $\Q\in \mathcal P$ is the 'correct' model, where $\alpha(\Q)=0$ corresponds to maximal confidence.\ The (monetary) risk measure $R$ can be interpreted as the outcome of some internal risk assessment.
Since, by assumption, $\inf_{\Q\in \mathcal P}\alpha (\Q)=0$, it follows that $$\big\{\PD_\Q(X)\,\big|\, \Q\in \mathcal P,\, \alpha(\Q)\leq R(X)\big\}\neq\emptyset\quad\text{for all }X\in \Bb\text{ with }R(X)>0.$$ In particular, $\DR(X)\in [0,1]$ for all $X\in \Bb$ and $\DR(m)=1$ for all $m\in \R$ with $m>0$. For $X,Y\in \Bb$ with $X\leq Y$, it follows that $ R(X)\leq R(Y)$ and $\PD_\Q(X)\leq \PD_\Q(Y)$ for all $\Q\in \mathcal P$, so that $\DR(X)\leq \DR(Y)$.
\end{example}

\JS{After some examples have been discussed, we now turn our focus on extensions of a given default risk measure $\DR\colon C\to [0,1]$ from the set of existing customers $C$ to the space $\Bb$ of all bounded measurable functions. We start with following definition.}

\begin{definition}\
Let $\DR\colon C\to [0,1]$ be a default risk measure and $F\colon \Bb\to \R$. 
\begin{enumerate}
\item[a)] We say that $F$ is \textit{monotone} if $F(X)\leq F(Y)$ for all $X,Y\in \Bb$ with $X\leq Y$.
\item[b)]
We say that $F$ is \textit{compatible} with $\DR$ if $F(0)=0$ and, for $X,Y\in C$, $F(X-Y)\leq 0$ implies that $\DR(X)\leq \DR(Y)$.
\end{enumerate}
\end{definition}

\begin{theorem}\label{thm.extension}
Let $\DR\colon C\to [0,1]$ be a default risk measure and $F\colon \Bb\to \R$ be monotone and compatible with $\DR$. Then, $\DR_F\colon \Bb\to \R$, given by 
\[
\DR_F(X):=\inf\big\{\DR(X_0)\, \big|\, X_0\in C,\, F(X-X_0)\leq 0\big\}\quad\text{for all }X\in \Bb,
\]
defines a default risk measure on $\Bb$ with $\DR_F(X)=\DR(X)$ for all $X\in C$. 
\end{theorem}

The previous theorem shows that any default risk measure $\DR$ on $C$ can be extended to a default risk measure on $\Bb.$ \JS{Assume that a financial institution has a set $C$ of existing customers and a default risk measure $\DR$, which represents a rating system, assigning to each customer a rating class in the form of a probability of default.\ If the financial institution now aims to allocate new customers, which are characterized by different or additional risk factors, within the existing rating classes, this is possible with the extension procedure described in Theorem \ref{thm.extension}.\ Due to the construction of $\DR_F$, no new rating classes are introduced and the monotone function $F\colon \Bb\to \R$ might resemble an internal risk measurement procedure. Here, the compatibility condition, which may not seem very intuitive at first, entails a comparison in the risk assessment between new and existing customers.}

A particularly interesting choice for a monotone map $F\colon \Bb\to \R$, as in Theorem \ref{thm.extension}, is given by the choice $F(X):=\sup X$ for all $X\in \Bb$.\ This leads to the default risk measure
$$
\DR_{\sup}(X):=\inf\big\{\DR(X_0)\, \big|\, X_0\in C,\, X\leq X_0\big\}\quad\text{for all }X\in \Bb,
$$
which is akin to the idea of superhedging.\ Note that this choice of $F$ is compatible with \textit{every} default risk measure $\DR\colon C\to [0,1]$. \JS{Furthermore, this choice of an extension leads to the most conservative default risk measure which is consistent with $\DR$ as the following corollary indicates.}

\begin{corollary} \label{cor. est}
Let $\DR\colon C\to [0,1]$ be a default risk measure.\ Then, for every default risk measure $\overline \DR\colon \Bb\to [0,1]$ with $\overline \DR(X_0)=\DR(X_0)$ for all $X_0\in C$, it holds
\[
\overline \DR(X)\leq \DR_{\sup}(X)\quad\text{for all }X\in \Bb.
\]
\end{corollary}

Another interesting example is given by the case, where $C$ consists of only constant functions and $F$ is a monetary risk measure, e.g., the expected value with respect to a probability measure $\P$ on $\FF$. This case is discussed in the following example.

\begin{example}[Binary default risk measure]\label{ex.mean_value} 
Let $R\colon \Bb\to \R$ be a monetary risk measure. For $X\in \Bb$, let
\[
\DR(X):=\inf\big\{\eins_{(0,\infty)}(m)\, \big|\, m\in \R,\, R(X-m)\leq 0\big\}.
\]
Then, for all $X\in \Bb$,
 $$ \DR(X)= \begin{cases}
  0, & \text{if } R(X) \leq 0, \\
  1, & \text{otherwise}.
\end{cases}$$
Since $\DR$ only takes two values, it is an almost trivial example of a default risk measure, which, however, has some surprising properties in terms of acceptance sets and its related value at risk, cf. Section \ref{sec.acceptance}, below.
\end{example}

\section{Additional properties and value at risk}\label{sec.acceptance}

Recall that a monetary risk measure $R\colon \Bb\to \R$ is uniquely determined by its acceptance set
$$\sA_R:=\big\{ X\in \Bb\, \big|\, R(X)\leq0\big\}$$  
and, given a probability measure $\P$ on $(\Omega, \mathcal F)$, the \textit{value at risk} (VaR) at level $\alpha\in(0,1)$ is given by 
$$\VaR^\alpha_\P (X) := \inf \big\{ m\in \R\,\big|\,\P(X-m>0)\leq\alpha\big\}.$$
For a default risk measure $\DR\colon \Bb\to [0,1]$, the condition $\DR(X)\leq0$ implies that $X\in \Bb$ exhibits no default risk whatsoever, so that an acceptance set similar to the one of a monetary risk measure is not meaningful for default risk measures.\ In this context, it is worth noting that default risk is not expressed in monetary units.\ In particular, default risk measures are not cash additive.\ Having the PD with respect to a probability measure $\P$ in mind as a prime example for a default risk measure, and looking at its connection to the VaR, a different approach seems more natural.\ The aim of this section is to formalize this relation between PD and VaR, and transfer it to general default risk measures. In view of Theorem \ref{thm.extension}, we focus on the case $C=\Bb$ throughout this section, and start with the following definition, which is central for the subsequent discussion. 

\begin{definition}
Let $\DR\colon \Bb\to [0,1]$ be a default risk measure and $\alpha\in (0,1)$. For $X\in \Bb$, let $$\VR (X) := \inf \big\{ m\in \R\,\big|\,\DR(X-m)\leq\alpha\big\}. $$ Then, $\VR$ is called the $\DR$\textit{-value at risk} at level $\alpha$.
\end{definition}

For default risk measures, we consider the following additional properties.

\begin{definition}\label{def.addprop}
 Let $\DR\colon \Bb\to [0,1]$ be a default risk measure.
\begin{enumerate}
\item[a)] We say that $\DR$ is \textit{quasi-convex} if
\begin{equation*}
\DR(\lambda X +(1-\lambda)Y)\leq \DR(X) \lor \DR(Y)\quad\text{for all } \lambda\in [0,1] \text{ and } X,Y\in \Bb.
\end{equation*}
\item[b)] We say that $\DR$ is \textit{scaling invariant} if
\begin{equation}\label{eq.sufficient_condition}
 \JS{\DR(X)= \DR(\lambda X)\quad\text{for all }\lambda>0 \text{ and } X\in \Bb.}
\end{equation}
\end{enumerate}
\end{definition}
\JS{We point out that quasi-convexity and scaling invariance, i.e., positive homogeneity of degree zero, are well-known properties in theoretical economics and mathematical finance. In the context of credit risk, scaling invariance has appeared in \cite{wang2021rating} in a different setting under the name \textit{nominal-invariance}. In the economic literature, positive homogeneity of degree zero or, simply, \textit{homogeneity} is also known under the name \textit{scale independence}, see for example \cite{MR0573315,MR0771167}.}

\begin{proposition}\label{prop.VR}
Let $\DR\colon \Bb\to [0,1]$ be a default risk measure. 
\begin{enumerate}
  \item[a)] For all $\alpha \in (0,1)$, $\VR$ is a monetary risk measure, cf.\ Definition \ref{def.riskmeasure}.
\item[b)] $\DR$ can be recovered from the family $(\VR)_{\alpha\in (0,1)}$ via
\begin{equation}\label{eq.recovery}
\DR(X)=\inf \Big(\big\{\alpha\in (0,1) \,\big|\, \VR(X)\leq 0  \big\}\cup\{1\}\Big).
\end{equation}
\item [c)] $\DR$ is scaling invariant if and only if $\VR$ is positively homogeneous for all $\alpha\in (0,1)$.
\item [d)] $\DR$ is quasi-convex if and only if $\VR$ is convex for all $\alpha\in (0,1)$. 
\end{enumerate}
\end{proposition} 

\begin{remark}\label{rem.varconvex}
Although a default risk measure $\DR\colon \Bb\to \R$ is never convex, there are default risk measures that are quasi-convex. An example for such a default risk measure is given by
$$ \DR(X):= \begin{cases} 0,& \text{if }R(X)\leq 0,\\
1,&\text{else},
\end{cases}
\quad\text{for all }X\in \Bb,
$$
where $R\colon \Bb\to \R$ is a monetary risk measure in the sense of Definition \ref{def.riskmeasure}. Then, for $X\in \Bb$ and $m\in \R$, $\DR(X-m)\leq 0$ if and only if $R(X)\leq m$. Hence,
\[
\VR(X)=R(X)\quad\text{for all }\alpha\in (0,1).
\]
By Proposition \ref{prop.VR} d), $\DR$ is quasi-convex if and only if $R$ is convex.\ Moreover, by Proposition \ref{prop.VR} c), $\DR$ is scaling invariant if and only if $R$ is positively homogeneous.
\end{remark}

Thinking of the family $(\VR)_{\alpha\in (0,1)}$ for a default risk measure $\DR$, leads to the following definition.
\begin{definition}
A \textit{generalized quantile function} is a family $(R^\alpha)_{\alpha\in (0,1)}$, where
\begin{enumerate}[(i)]
\item $R^\alpha\colon \Bb\to \R$ is a monetary risk measure for each $\alpha\in (0,1)$, cf.\ Definition \ref{def.riskmeasure},
\item $R^\beta(X)\leq R^\alpha(X)$ for all $X\in \Bb$ and $\alpha,\beta\in(0,1)$ with $\alpha\leq \beta$,
\item for all $\alpha\in (0,1)$ and $X\in \Bb$, $$R^\alpha(X)=\sup_{\beta\in (\alpha,1)}R^\beta(X).$$
\end{enumerate}
\end{definition}

\begin{remark}
For each default risk measure $\DR\colon \Bb\to [0,1]$, the family $\big(\VR\big)_{\alpha\in (0,1)}$ is a generalized quantile function.\ The properties (i) and (ii) are immediate consequences of Proposition \ref{prop.VR} a) and the definition of the family $\big(\VR\big)_{\alpha\in (0,1)}$, respectively.\ We prove (iii) by contradiction.\ To that end, let $\alpha\in (0,1)$ and $X\in \Bb$, and suppose that
\[
\VR(X)>\sup_{\beta\in (\alpha,1)}\VaR_\DR^\beta(X).
\]
Then, there exists some $m\in \R$ with $$\VR(X)>m>\VaR_\DR^\beta(X) \quad \text{for all }\beta\in (\alpha,1),$$
which implies that $\alpha<\DR(X-m)\leq \beta$ for all $\beta\in (\alpha,1)$, and thus leads to contradiction.
\end{remark}
The following theorem together with Proposition \ref{prop.VR} b) and the previous remark shows that there is a one-to-one relation between default risk measures and generalized quantile functions.

\begin{theorem}\label{thm.recalculate}
Let $(R^\alpha)_{\alpha\in (0,1)}$ be a generalized quantile function and
$$ \DR(X):=\inf \Big(\big\{\alpha\in (0,1) \,\big|\, R^\alpha(X)\leq 0 \big\}\cup\{1\}\Big).$$
Then, $\DR\colon \Bb\to[0,1]$ is a default risk measure and $R^\alpha=\VR$ for all $\alpha\in (0,1)$.
\end{theorem}

\begin{proposition} \label{Prop. VaR}
    Let $\DR\colon \Bb\to [0,1]$ be a default risk measure and $\mathcal P$ be a nonempty set of probability measures. Then, the following are equivalent:
    \begin{enumerate}[(i)]
    \item For all $X\in \Bb$,
    \[
    \DR(X)=\sup_{\P\in \mathcal P}\P(X>0).
    \]
    \item For all $X\in \Bb$ and all $\alpha\in (0,1)$,
    \begin{equation}\label{eq.robustvar}
    \VR(X)=\sup_{\P\in \mathcal P}\VaR^\alpha_\P(X).
    \end{equation}
    \end{enumerate}
\end{proposition}

\begin{remark}\label{rem.comonotone}\
\begin{enumerate}
\item[a)]  In \cite[Proposition 4.6]{MR4318899}, it is shown that the worst-case value at risk, given by \eqref{eq.robustvar}, is \textit{comonotonically additive}. More generally, \cite[Corollary 1]{MR2888471} implies that $\VR$ is comonotonically additive if $\DR(X)=\DR\big(\eins_{\{X>0\}}\big)$ for all $X \in\Bb $, i.e., $\DR$ is given by a capacity $c\colon \FF \to [0,1]$, cf. Section \ref{sc:Choquet} below.
In fact, \cite[Corollary 1]{MR2888471} states that, for every $\alpha\in(0,1)$, there exists a capacity $c_\alpha\colon \FF\to [0,1]$ such that $\VR$ is given by the Choquet integral w.r.t.\ $c_\alpha$, i.e.,
\begin{equation}\label{eq.chambers}
	\VR(X)= \int X\, {\rm d} c_\alpha \quad \text{for all }X\in\Bb,
\end{equation}
see Section \ref{sc:Choquet} below for a definition of the Choquet integral.\ Clearly, if $\VR$ satisfies \eqref{eq.chambers}, it is positively homogeneous for all $\alpha\in (0,1)$, so that $\DR$ is scaling invariant by Proposition \ref{prop.VR} c).\ The natural question arises if a representation of $\VR$ in terms of \eqref{eq.chambers} for all $\alpha\in (0,1)$ already implies that $\DR$ is given by a capacity as well. We provide a negative answer to this question in part b).

\item[b)] Let $\P$ be a probability measure on $(\Omega,\FF)$, and assume that there exists some event $A\in \FF$ with $p:=\P(A)\in (0,1)$.\ Choosing, in the situation of Remark \ref{rem.varconvex}, $R=\E_\P$ as the expected value under $\P$, it follows that $\VR=\E_\P$, so that $\VR$ satisfies \eqref{eq.chambers} for all $\alpha\in (0,1)$. However, for $X:=\frac1p \eins_A-\frac{1}{1-p}\eins_{A^c}$, it follows that
\[
\E_\P(X)=0\neq 1=\E_\P(X^+),
\]
so that $\DR(X)=0\neq 1= \DR(X^+)$.\ Hence, by Theorem \ref{thm.main} below, $\DR$ is \textit{not} given by a capacity.
\end{enumerate}
\end{remark}

\begin{proposition} \label{Prop. robust VaR}
    Let $\DR\colon \Bb\to [0,1]$ be a default risk measure, $\P$ a probability measure, and $T\colon [0,1]\to [0,1]$ a nondecreasing and lower semicontinuous function with $T(0)=0$ and $T(1)=1$. Then, the following are equivalent:
    \begin{enumerate}[(i)]
    \item For all $X\in \Bb$,
    \[
    \DR(X)=T\big(\P(X>0)\big).
    \]
    \item For all $X\in \Bb$ and all $\alpha\in (0,1)$,
    \begin{equation}\label{eq.distortedvar}
    \VR(X)=\VaR^{T^{-1}(\alpha)}_\P(X),
    \end{equation}
    where $T^{-1}(a):=\inf\{b\in (0,1)\, |\, T(b)> a\}$ for all $a\in (0,1)$.
    \end{enumerate}
\end{proposition}

\JS{We conclude this section with two examples for generalized quantile functions arising in the context of regulatory risk measures and their resulting default risk measures.}

\begin{example}\
\JS{Let $\P\colon \FF\to [0,1]$ be a probability measure.}
\begin{enumerate}
 \item[a)] \JS{We consider the \textit{expected shortfall} $$\ES^\alpha_\P(X):=\frac1\alpha \int_0^\alpha \VaR^\lambda_\P(X)\, \d \lambda \quad\text{of }X\in \Bb \text{ at level }\alpha\in (0,1).$$ Then, $\ES^\alpha_\P\colon \Bb\to \R$ defines a monetary risk measure for each $\alpha\in (0,1)$, and $\ES^\beta_\P(X)\leq \ES^\alpha_\P(X)$ for all $X\in \Bb$ and $\alpha, \beta\in (0,1)$ with $\alpha\leq \beta$. Moreover, by the monotone convergence theorem, Property (iii) in the definition of a generalized quantile function is satisfied. Then, the default risk measure $\DR\colon \Bb\to [0,1]$ related to the generalized quantile function $(\ES_\P^\alpha)_{\alpha\in (0,1)}$ is given by
 \[
 \DR(X):=\inf \Big(\big\{\alpha\in (0,1) \,\big|\, \ES_\P^\alpha(X)\leq 0 \big\}\cup\{1\}\Big)\quad\text{for all }X\in \Bb.
 \]
 Let $X\in \Bb$ with $X\neq 0$. First, observe that $$\DR(X)=1\quad \text{if and only if}\quad \E_\P(X)=\int_0^1\VaR_\P^\lambda(X)\, \d \lambda\geq 0.$$ If $\E(X)<0$, then the continuity of the map $[0,1]\to \R,\;\alpha\mapsto \int_0^\alpha \VaR_\P^\lambda(X)\, \d \lambda$ implies that
 \begin{align*}
\DR(X)&=\inf\bigg\{\alpha\in (0,1)\, \bigg|\, \frac{1}{\alpha}\int_0^\alpha \VaR_\P^\lambda(X)\, \d \lambda\leq 0\bigg\}\\
&=\inf\bigg\{\alpha\in [0,1)\, \bigg|\, \int_0^\alpha \VaR_\P^\lambda(X)\, \d \lambda=0\bigg\}.
 \end{align*}
 In particular, by Theorem \ref{thm.recalculate} and Proposition \ref{Prop. VaR} with $\mathcal P=\{\P\}$, $$\DR(X)\geq \P(X>0)=\PD_\P(X)\quad\text{for all }X\in \Bb,$$ i.e., $\DR$ is a more conservative default risk measure than $\PD_\P$.}
 \item[b)] \JS{For $X\in \Bb$ and $\alpha\in (0,1)$, we consider the \textit{$\alpha$-expectile}, given by
 $$\Exp_\P^\alpha(X):=\inf\big\{m\in \R\, \big|\, (1-\alpha)\E_\P[(X-m)^-]- \alpha \E_\P[(X-m)^+]\leq 0\big\}.$$
 Then, $\Exp_\P^\alpha\colon \Bb\to \R$ is a monetary risk measure for all $\alpha\in (0,1)$.\ By the monotone convergence theorem, for all $X\in \Bb$,
 \[
 \alpha \E_\P\big[\big(X-\Exp_\P^\alpha(X)\big)^+\big]= (1-\alpha) \E_\P\big[\big(X-\Exp_\P^\alpha(X)\big)^-\big].
 \] 
 In particular, $\Exp^{1/2}_\P(X)=\E_\P(X)$. Since the map $(0,1)\to (0,\infty),\, \alpha\mapsto \frac{1-\alpha}\alpha$ is decreasing and continuous, the properties (ii) and (iii) of a generalized quantile function follow.\ Let $X\in \Bb$ and 
\[
 \DR(X):=\inf \Big(\big\{\alpha\in (0,1) \,\big|\, \Exp_\P^\alpha(X)\leq 0 \big\}\cup\{1\}\Big).
 \]
Then, $\DR(X)=1$ if and only if $X\geq 0$ $\P$-a.s and $\P(X>0)>0$. Moreover, $\DR(X)=0$ if and only if $X\leq 0$ $\P$-a.s.\ In all other cases, i.e., if $\DR(X)\in (0,1)$, then $\DR(X)$ is given by the unique solution $\DR(X):=\alpha^*\in (0,1)$ to the equation}
\[
\JS{\alpha^* \E_\P(X^+)=(1-\alpha^*)\E_\P(X^-).}
\]
\end{enumerate}
\end{example}

\section{Representation via capacities}\label{sec.properties}
After discussing first characteristics, generalizations of the $\VaR$, and different examples for default risk measures, we now aim for equivalent conditions for default risk measures to admit a representation of the form
\begin{equation}\label{eq.representation}
\DR(X)=\DR\big(\eins_{\{X>0\}})\quad\text{for all }X\in \Bb.
\end{equation}
This representation allows to describe the default risk measure $\DR$ via its related capacity, given by $c(A):=\DR(\eins_A)$ for $A\in \FF$. This description allows to derive explicit representations of $\DR$ as worst-case PDs and distorted PDs in Section \ref{sc:Choquet} and Section \ref{sec:distorted}, respectively.\ 

\JS{In the following, we discuss a slightly weaker notion of scaling invariance for default risk measures.}
\begin{definition}\label{def.dfscalinginv}
\JS{We say that a default risk measure $\DR\colon \Bb\to \R$ is \textit{default scaling invariant} if}
\begin{equation*}
\JS{\DR(X)= \DR(\lambda X)\quad\text{for all }\lambda>0 \text{ and } X\in \Bb\text{ with }X\geq 0.}
\end{equation*}
\end{definition}
Default risk measures that are not \JS{default} scaling invariant can depend, for example, on the maximum possible default.\ This applies, for example to the default risk measure in Example \ref{WS}, where a warning signal $\gamma$ must be exceeded by $\sup X$ in order to switch to a more conservative default risk measure. Another example for a default risk measure that is not default scaling invariant is given in Example \ref{ex.increasing}.

From an economic point of view, \JS{default} scaling invariance means that the default risk does not change if all potential losses, i.e., $X(\omega)>0$ are multiplied with a positive constant. Hence, for a given customer $X\geq 0$, we are not interested how high this customer defaults, which would be the case in Example \ref{WS}, but rather if they default or not. All scenarios in which the customer cannot pay a due at the end of the observation period are equivalent to the scenarios in which the due multiplied by $\lambda >0$ cannot be paid. Certainly, therefore, it is a necessary condition in order to obtain the desired representation \eqref{eq.representation}.

\begin{remark} \label{rk:scaling}
 Before we provide sufficient and necessary conditions for $\DR$ to have the form \eqref{eq.representation}, we first discuss various properties of $\DR$ that will appear at a later stage in this section.
\begin{enumerate}[a)]
\item Let $\DR\colon \Bb\to [0,1]$ be a \JS{default} scaling invariant default risk measure. Then, for all $X\in \Bb$,
\begin{equation}\label{eq.conservative}
\DR(X)\leq \DR\big(\eins_{\{X>0\}}\big),
\end{equation}
Indeed, for $X\in \Bb$ with $X\leq 0$,
$\DR(X)\leq \DR(0)=\DR\big(\eins_{\{X>0\}}\big)$.\
On the other hand, if $X\in \Bb$ with $\sup X>0$, the \JS{default} scaling invariance of $\DR$ yields that
 \[
 \DR(X)\leq \DR(X^+)\leq \DR\big((\sup X) \eins_{\{X>0\}}\big)= \DR\big(\eins_{\{X>0\}}\big).
 \]
Condition \eqref{eq.conservative} states that, under the mild assumption of \JS{default} scaling invariance, the most conservative choice of a default risk measure $\DR$ is a representation via a capacity. 
 \item \JS{To see that scaling invariance, in particular, default scaling invariance is not enough to come up with the representation $\DR(X)=\DR\big( \eins_{\{X>0\}}\big)$ for all $X\in \Bb$, we again pick up Example \ref{ex.mean_value}, where a default risk measure was defined by means of a monetary risk measure $R$.} Let $\P$ be a probability measure on $\mathcal F$, and consider the case $R(X)=\E_\P(X)$ for all $X\in \Bb$, where $\E_\P(\,\cdot\,)$ denotes the mean value under $\P$. \JS{Then, the default risk measure in Example \ref{ex.mean_value} is scaling invariant by Remark \ref{rem.varconvex}. However, the default risk measure very much depends on both the positive and negative values of $X\in \Bb$ as soon as there exists an event $A\in \FF$ with $\P(A)\in (0,1)$. In fact, let $A\in \FF$ with $p:=\P(A)\in (0,1)$ and define
 \[
 X:=\eins_A-\frac{2p}{1-p}\eins_{A^c}.
 \]
Then,
\[
 \E_\P(X)=\P(A)-\frac{1+p}{1-p}\cdot \P(A^c)=-1<0,
\]
so that $\DR(X)=0$. On the other hand, $\P(X>0)=\P(A)=p>0$, which implies that $\DR\big(\eins_{\{X>0\}}\big)=1$.} Nevertheless, in Proposition \ref{prop.condition} below, we will see that, up to certain continuity requirements, \JS{default} scaling invariance implies the equality $\DR(X)=\DR\big( \eins_{\{X>0\}} \big)$ for all $X\in \Bb$ with $X\geq 0$. \JS{Later we seek for additional properties of $\DR$ such as $\DR(X)=\DR(X^+)$ for all $X\in \Bb$.}
\item In order to explain why additional continuity assumptions are needed in Proposition \ref{prop.condition} below, we regard, for a given probability measure $\P$ on $\FF$, the default risk measure $\DR\colon \Bb\to [0,1]$, given by
$$
\DR(X):= \begin{cases}
  1,& \text{if } \P(\{X>\varepsilon\})=1 \text{ for some }\varepsilon>0,\\
  \frac{\P(\{X>0\})}{2}, & \text{otherwise.}  
\end{cases}
$$
 As such $\DR$ is \JS{default} scaling invariant and does not depend on the part, where $X\in \Bb$ is less or equal to $0$.\ However, if $(\Omega,\FF,\P)$ is atomless and $X\in \Bb$ is uniformly distributed on $[0,1]$ under $\P$, then
$$
\P\big(X\in (0,\varepsilon)\big)=\varepsilon>0\quad\text{for all }\varepsilon>0,
$$
so that
$$\DR(X)=\frac{1}{2}\neq 1=\DR\big( \eins_{\{X>0\}} \big).$$ 
The following proposition clarifies this \JS{elementary} example.
 \end{enumerate}
\end{remark}

\begin{proposition}\label{prop.condition}
Let $\DR\colon \Bb\to [0,1]$ be a \JS{default} scaling invariant default risk measure.
\begin{enumerate}[a)]
\item For all $X\in \Bb$ with $X\geq 0$,
\begin{equation}\label{eq.capacity_below}
\DR(X)\leq\DR\big(\eins_{\{X>0\}}\big)\leq \inf_{\varepsilon >0} \DR\big(X+\varepsilon\eins_{\{X>0\}}\big).
\end{equation}
\item For all $X\in \Bb$ with $X\geq 0$,
\begin{equation}\label{eq.quasi_capacity}
 \sup_{\varepsilon >0} \DR\big(\eins_{\{X>\varepsilon\}}\big)\leq \DR(X)\leq \DR\big(\eins_{\{X>0\}}\big).
\end{equation}
\end{enumerate}
\end{proposition}

The previous proposition shows that \JS{default} scaling invariance together with $\DR(X)=\DR(X^+)$, for all $X\in \Bb$, implies that the default risk measure can be represented  via a capacity except for some additional continuity properties.\ The inequality in part a) is quite mild from a theoretical perspective, because reducing a customer's cash flows by an amount of money $\varepsilon>0$ in scenarios, where the customer is defaulting anyways, has no impact on the default scenarios. In contrast to this, the continuity assumption 
$$
\sup_{\varepsilon >0} \DR\big(\eins_{\{X>\varepsilon\}}\big)=\DR\big(\eins_{\{X>0\}}\big)\quad\text{for all }X\in \Bb
$$
is not for free since, in mathematical terms, it requires some mild form of continuity from below for the related capacity.\ From an economic point of view, however, it is quite intuitive since $\varepsilon>0$ can be chosen arbitrarily small, e.g., in such a way that it falls below the smallest possible amount of money. A theoretical increase in capital by $10^{-10}$ euros should not change the default risk of the respective customer. This motivates to introduce the following definition.

\begin{definition}\label{def.illiquidity}
Let $\DR\colon \Bb\to [0,1]$ be a default risk measure. We say that $\DR$ is \textit{illiquidity invariant}, if, for all $m>0$ and $X\in \Bb$ with $X\geq 0$,
$$\DR(X)=\DR\big(X+m\eins_{\{X>0\}}\big).$$
\end{definition}

In other words, \textit{illiquidity invariance} means that customers that are in default or, loosely speaking, illiquid remain illiquid if their debt due increases or their capital decreases in the sense that their default risk does not change. \JS{The following proposition establishes a first connection between illiquidity invariance, default scaling invariance, and a representation via indicator functions.}

\begin{proposition}\label{prop.capacity_represenatation}
Let $\DR\colon \Bb\to [0,1]$ be a default risk measure.\ Then, the following statements are equivalent.
\begin{enumerate}[(i)]
\item $\DR$ is \JS{default} scaling invariant and $$\DR(X)=\inf_{\varepsilon>0}\DR(X+\varepsilon \eins_{\{X>0\}})\quad\text{for all }X\in \Bb\text{ with }X\geq 0.$$
\item $\DR$ is illiquidity invariant.
\item For all $X\in \Bb$ with $X\geq 0$,
$$\DR(X)=\DR\big(\eins_{\{X>0\}}\big).$$
\end{enumerate}
\end{proposition}

In analogy to Definition \ref{def.illiquidity}, we speak of \textit{liquidity invariance}, when the default risk remains the same if capital is added in scenarios where the customer is able to make all payments due. 

\begin{definition}
Let $\DR\colon \Bb\to [0,1]$ be a default risk measure. We say that $\DR$ is \textit{liquidity invariant}, if, for all $m>0$ and $X\in \Bb$ with $X\geq 0$,
$$\DR(X)=\DR\big(X-m\eins_{\{X=0\}}\big).$$
\end{definition}

\begin{lemma}\label{lem.X_plus}
Let $\DR\colon \Bb\to [0,1]$ be a default risk measure. Then, the following conditions are equivalent.
\begin{enumerate}[(i)]
\item $\DR(X)=\DR(X^+)$ for all $X\in \Bb$.
\item $\DR(X)=\DR\big(X-m\eins_{\{X\leq 0\}}\big)$ for all $m>0$ and $X\in \Bb$.
\item $\DR$ is liquidity invariant.
\end{enumerate}
\end{lemma}

In the context of risk functionals, Property (i) in Lemma \ref{lem.X_plus} (with a different sign convention) is known under the names \textit{surplus invariance}, cf.\ \cite{MR4166752}, \textit{loss dependence}, cf.\ \cite{MR3081794}, and \textit{excess invariance}, cf.\ \cite{MR3017383}.\ A combination of Proposition \ref{prop.capacity_represenatation} and Lemma \ref{lem.X_plus} leads to the following theorem, which is the main result of this section. 

\begin{theorem} \label{thm.main} Let $\DR\colon \Bb\to [0,1]$ be a default risk measure. Then, the following conditions are equivalent.
\begin{enumerate}[(i)]
\item For all $X\in \Bb$, $$\DR(X)=\DR\big( \eins_{\{X>0\}} \big).$$
\item $\DR$ is \JS{default} scaling invariant with
\[
\DR(X)=\DR(X^+)=\inf_{\varepsilon >0}\DR\big(X+\varepsilon \eins_{\{X>0\}} \big)\quad\text{for all }X\in \Bb.
\]
\item $\DR$ is liquidity invariant and illiquidity invariant.
\end{enumerate}
\end{theorem}

To get a link to Section \ref{sc:Choquet}, where the focus lies on the related capacities and their Choquet integrals,  we introduce the notion of \textit{submodularity} for default risk measures. In theoretical economics, submodularity is a classical property, which is closely related to substitute goods, cf.\ \cite{MR1614637}.

\begin{definition}
Let $C$ be a sublattice of $\Bb$, which contains all constant functions.\ A default risk measure $\DR\colon C\to [0,1]$ is called \textit{submodular} if
$$\DR(X\land Y)+\DR(X\lor Y) \leq \DR(X)+\DR(Y)\quad\text{for all }X,Y\in C.$$
\end{definition}

The natural question arises, whether submodularity is a sensible notion for default risk measures. For classical PDs, we have equality, i.e.,
$$\PD_\P(X\land Y)+\PD_\P(X\lor Y)=\PD_\P(X)+\PD_\P(Y)\quad\text{for all }X,Y\in \Bb.$$
Moreover, submodularity is given for certain classes of distorted PDs, for instance, if the distortion function is concave, see Section \ref{sec:distorted}. 

\begin{remark}\label{rem.submoldular}
Let $C$ be a sublattice of $\Bb$, which contains all constant functions, and $\DR\colon C\to [0,1]$ be a submodular default risk measure. Then,
\[
\DR(X)=\DR(X^+)\quad\text{for all }X\in C.
\]
In fact, let $X\in C$. The submodularity of $\DR$ together with $0\in C$ implies that
\[
\DR(X^+)+\DR(-X^-)\leq \DR(X).
\]
Since $-X^-\leq 0$, it follows that $\DR(-X^-)=0$ by the defining properties of a default risk measure. Therefore,
\[
\DR(X)\leq \DR(X^+)=\DR(X^+)+\DR(-X^-)\leq \DR(X).
\]
\end{remark}

As a consequence of Theorem \ref{thm.main} and Remark \ref{rem.submoldular}, we obtain the following corollary.

\begin{corollary} \label{cor.sublodular} 
Let $\DR\colon \Bb\to [0,1]$ be a submodular default risk measure.\ Then, the following statements are equivalent.
\begin{enumerate}[(i)]
\item $\DR$ is illiquidity invariant.
\item  For all $X\in \Bb$, $$\DR(X)=\DR\left( \eins_{\{X>0\}} \right).$$
\end{enumerate}
\end{corollary}

\section{Choquet integrals and robust representations} \label{sc:Choquet}

In Section \ref{sec.properties}, we discussed equivalent conditions for a default risk measure $\DR\colon \Bb\to [0,1]$ in order to be the form \eqref{eq.representation}. Building on this representation, the aim of this section is to connect default risk measures with Choquet integrals and monetary risk measures \JS{and use this connection to attain a representation for $\DR$ via probability measures}. Recall that a \textit{capacity} is a map $c\colon \FF\to [0,1]$ with $c(\emptyset)=0$, $c(\Omega)=1$, and $c(A)\leq c(B)$ for all $A,B\in \FF$ with $A\subset B$. We start with the following observation.

\begin{remark}\label{rem.choquet}
Let $\DR\colon \Bb\to [0,1]$ be a default risk measure.\ Then, we can define a capacity $c\colon \FF \to [0,1]$ by \begin{equation}\label{eq.capacity}
c(A):=\DR\left(\eins_{A}\right)\quad\text{for all }A\in \FF.
\end{equation}
By definition of a default risk measure,
$$c(\emptyset)=c(\eins_\emptyset)=\DR(0)=0 \quad\text{and}\quad c(\Omega)=c(\eins_\Omega)=\DR(1)=1.$$ Moreover, for all $A,B\in \FF$ with $A\subset B$, 
$$c(A)=\DR\left(\eins_{A}\right) \leq \DR\left(\eins_{B}\right)=c(B)$$
due to the monotonicity of $\DR$. For $X\in \Bb$, the \textit{Choquet integral} with respect to $c$ is defined as 
$$\int X \, \d c:=\int_{-\infty}^{0}\big(c( \{X>s\})-1 \big)\, \d s+\int_{0}^{\infty}c(\{X>s\})\,\d s.$$
Although the Choquet integral is, in general, not a linear functional, it defines a monetary risk measure $R\colon \Bb\to \R$ via
\[
R(X):=\int X\, \d c\quad\text{for all }X\in \Bb.
\]
By definition of the Choquet integral, the monetary risk measure $R$ is positively homogeneous, i.e., $R(\lambda X)=\lambda R(X)$ for all $X\in \Bb$ and $\lambda>0$. If $\DR$ satisfies \eqref{eq.representation}, then
\[
\DR(X)=R\big(\eins_{\{X>0\}}\big)\quad\text{for all }X\in \Bb,
\]
which leads to the following proposition.
\end{remark}

\begin{proposition} \label{prop. sc 5}
    Let $\DR\colon \Bb\to [0,1]$ be a default risk measure. Then, the following two statements are equivalent.
\begin{enumerate}[(i)]
\item For all $X\in \Bb$, 
\[
 \DR(X)=\DR\big(\eins_{\{X>0\}}\big).
\]
 \item There exists a positively homogeneous \JS{monetary} risk measure $R\colon \Bb\to \R$\JS{, cf.\ Definition \ref{def.riskmeasure},} with
 \[
 \DR(X)=R\big(\eins_{\{X>0\}}\big)\quad \text{for all }X\in \Bb.
 \]
\end{enumerate}
\end{proposition}

 Let $\DR\colon \Bb\to [0,1]$ be a default risk measure, which satisfies \eqref{eq.representation}. Then, $\DR$ is submodular if and only if the related capacity $c\colon \FF\to [0,1]$, given by \eqref{eq.capacity}, is \textit{2-alternating}, i.e.,
 \[
 c(A\cup B)+c(A\cap B)\leq c(A)+c(B)\quad\text{for all }A,B\in \FF.
 \]
 A well-known fact is that a capacity is 2-alternating if and only if the related Choquet integral defines a coherent risk measure.\ We recall that a monetary risk measure $R\colon \Bb\to \R$ is coherent if and only if there exists a nonempty set $\cP$ of finitely additive probability measures such that
\[
R(X)=\sup_{\Q\in \cP}\E_\Q(X)\quad\text{for all }X\in \Bb,
\]
cf.\ \cite{MR3859905}. In this case, we obtain the robust representation
\begin{equation}\label{eq.integralrep}
\DR(X)=\int \eins_{\{X>0\}}\,{\rm d}c=\sup_{\Q\in \cP}\E_\Q( \eins_{\{X>0\}})=\sup_{\Q\in \cP}\Q(X>0)\quad\text{for all }X\in \Bb.
\end{equation}
We now investigate additional continuity properties for $\DR$ that guarantee a representation via countably additive probability measures.\ In the sequel, for a sequence $(X_n)_{n\in \N}\subset \Bb$ and $X\in \Bb$, we write $X_n\nearrow X$ or $X_n\searrow X$ as $n\to \infty$ if $X_n\leq X_{n+1}$ or $X_n\geq X_{n+1}$ for all $n\in \N$ and $X(\omega)=\lim_{n\to \infty}X_n(\omega)$ for all $\omega\in \Omega$, respectively.

\begin{proposition}\label{prop.contbelow}
    Let $\DR\colon \Bb\to [0,1]$ be a default risk measure with $\DR(X)=\DR\left( \eins_{\{X>0\}} \right)$ for all $X\in \Bb$ and $c(A):=\DR(\eins_A)$ for all $A\in \FF$.\ Then, the following statements are equivalent.
    \begin{enumerate}[(i)]
     \item $\DR$ is \textit{continuous from below}, i.e., for every sequence $(X_n)_{n\in \N}\subset \Bb$ with $X_n \nearrow X\in \Bb$ as $n\to \infty$,
$$
\DR(X)=\lim_{n\to \infty} \DR\left( X_n \right).
$$
\item For every sequence $(A_n)_{n\in \N}\subset \FF$ with $A_n\subset A_{n+1}$ for all $n\in \N$ 
$$
c\bigg(\bigcup_{n\in \N}A_n\bigg)=\lim_{n\to \infty} c(A_n).
$$
\item For every sequence $(X_n)_{n\in \N}\subset \Bb$ with $X_n \nearrow X\in \Bb$ as $n\to \infty$,
$$
\int X\, \d c=\lim_{n\to \infty} \int X_n\, \d c.
$$
 \end{enumerate}
\end{proposition}

Unfortunately, even if $\DR$ is submodular, continuity from below of $\DR$ is, in general, not a sufficient condition in order to guarantee a representation via countably additive probability measures on $\Bb$ \JS{as the following example shows}.

\begin{example}
\JS{Let $\P\colon \mathcal B\to [0,1]$ be the Lebesgue measure defined on the Borel $\sigma$-algebra $\mathcal B$ of the closed interval $\Omega:=[0,1]$ and $\FF$ be the power set of $\Omega$. Let $\mathcal L$ denote the space of all bounded $\mathcal B$-measurable functions $\Omega\to \R$ and
\[
R(X):=\inf\{\E_\P(X_0) |\, X_0\in \mathcal L, \, X_0\geq X\}.
\]
Then, by \cite[Proposition 2.2 and Lemma 3.6]{dkn}, $R$ is a coherent risk measure, which is continuous from below and the maximal extension of $\E_\P$ to $\Bb$.\ Consider the capacity $c\colon \FF\to [0,1]$, given by
\[
c(A):=R(\eins_A)\quad\text{for all }A\in \FF,
\]
and $\DR(X):=c(\{X>0\})$ for all $X\in \Bb$.\ Let $\Q\colon \FF\to [0,1]$ be a finitely additive probability with
\begin{equation}\label{eq.counterex}
\E_\Q(X)\leq R(X)\quad \text{for all }X\in \Bb.
\end{equation}
Then, for all $X\in \Bb$ with $X\geq 0$,
\[
\int X\, \d c= \int_0^\infty c(\{X>s\})\, \d s\geq \int_0^\infty \Q(\{X>s\})\, \d s=\E_\Q(X).
\]
Moreover, for $X\in \mathcal L$ with $X\geq 0$,
\[
\int X\, \d c= \int_0^\infty c(\{X>s\})\, \d s= \int_0^\infty \P(\{X>s\})\, \d s=\E_\P(X).
\]
Hence, by the maximality of $R$, it follows that $\int X\, \d c\leq R(X)$ for all $X\in \Bb$. We have therefore shown that $R(X)=\int X\, \d c$ for all $X\in \Bb$. In particular, $\DR$ is submodular. On the other hand, assuming the continuum hypothesis, by \cite[Satz 1C]{bierlein}, there exists not a single countably additive probability measure $\Q\colon \FF\to [0,1]$ with \eqref{eq.counterex}. In fact, with a similar construction, using a probability measure which is different from the Lebesgue measure, one can avoid invoking the continuum hypothesis, see \cite[Example 3.7]{dkn} for further details.}
\end{example}

If, however, $\DR$ is a distorted PD, which is submodular, continuity from below is a sufficient condition in order to allow for a representation in terms of countably additive probability measures as the following remark discusses.

\begin{remark}\label{rem.distpdcontbelow}
Let $\P$ be a probability measure of $\FF$ and $T\colon [0,1]\to [0,1]$ be a nondecreasing function with $T(0)=0$ and $T(1)=1$. Let $\DR\colon \Bb\to [0,1]$ be given by
\[
\DR(X)=T\big(\PD_\P(X)\big)\quad \text{for all }X\in \Bb,
\]
and assume that $\DR$ is submodular, so that $c\colon \FF\to [0,1],\; A\mapsto \DR(\eins_A)$ is a $2$-alternating capacity.
\begin{enumerate}[a)]
\item Assume that $T$ is left-continuous or, equivalently, lower semicontinuous, i.e.,
\[
T(p)=\sup_{q\in (0,p)}T(q)\quad\text{for all }p\in (0,1].
\]
Then, by Proposition \ref{prop.contbelow}, the Choquet integral w.r.t.\ $c$ is a coherent and law-invariant risk measure, which is continuous from below.\ Hence, by \cite[Theorem 4.33]{MR3859905}, there exists a set of countably additive and (w.r.t.\ $\P$) absolutely continuous probability measures $\cP$ on $\FF$ with
\[
\DR(X)=\sup_{\Q\in \cP}\PD_\Q(X)\quad\text{for all }X\in \Bb.
\]
\item[b)] Now, assume that $(\Omega,\FF,\P)$ is atomless.\ Then, the Choquet integral w.r.t.\ $c$ is a coherent and law-invariant risk measure on an atomless probability space.\ By \cite[Theorem 2.1]{MR2277714}, it follows that the Choquet integral w.r.t.\ $c$ is continuous from below. Hence, by \cite[Theorem 4.33]{MR3859905}, there exists a set of countably additive and (w.r.t.\ $\P$) absolutely continuous probability measures $\cP$ on $\FF$ with
\[
\DR(X)=\sup_{\Q\in \cP}\PD_\Q(X)\quad\text{for all }X\in \Bb.
\]
\end{enumerate}
\end{remark}

If $\Omega$ is a Polish space and $\mathcal F$ is the Borel $\sigma$-algebra, we have the following characterization for general submodular default risk measures on the space $\Lb$ of all bounded lower semicontinuous functions $\Om\to \R$.

\begin{proposition} \label{prop. polish}
Let $\Omega$ be a Polish space and $\DR\colon \Lb \to [0,1]$ be a submodular default risk measure with
\begin{equation}\label{eq.replb}
\DR(X)=\DR\big(\eins_{\{X>0\}}\big)\quad\text{for all }X\in \Lb\text{ with }X\geq 0.
\end{equation}
Then, the following statements are equivalent.
\begin{enumerate}[(i)]
\item $\DR$ is continuous from below, i.e.,
\[
\DR(X)=\lim_{n\to \infty}\DR(X_n)
\]
for every sequence $(X_n)_{n\in \N}\subset \Lb$ with $X_n\nearrow X\in \Lb$ as $n\to \infty$.
\item There exists a nonempty set $\cP$ of probability measures on $(\Omega,\FF)$ such that
\[
\DR(X)=\sup_{\Q\in \cP}\Q(X>0)\quad\text{for all }X\in \Lb.
\]
\end{enumerate}
\end{proposition}

\begin{remark}
Let $\DR\colon \Bb\to [0,1]$ be a default risk measure.\ Then, \textit{continuity from above} in the sense that
$$\DR(X)=\lim_{n\to \infty}\DR(X_n)$$
for every sequence $(X_n)_{n\in \N}\subset \Bb$ with $X_n\searrow X\in \Bb$ as $n\to \infty$, is \textit{not} a sensible property.
Considering a sequence $\left(X_n \right)_{n\in \N} \subset \Bb$ defined by $X_n=\frac{1}{n}$, it follows that $\DR(X_n)=1$, while $\DR(0)=0$.\ However, by \eqref{eq.integralrep}, it is enough to express the Choquet integral via a set of (\JS{countably} additive) probability measures. We can therefore weaken the requirement of continuity from above and obtain the following proposition.
\end{remark}

\begin{proposition}\label{prop.contabove}
    Let $\DR\colon \Bb\to [0,1]$ be a default risk measure with $\DR(X)=\DR\left( \eins_{\{X>0\}} \right)$ for all $X\in \Bb$ and $c(A):=\DR(\eins_A)$ for all $A\in \FF$.\ Then, the following statements are equivalent.
    \begin{enumerate}[(i)]
     \item For every sequence $(X_n)_{n\in \N}\subset \Bb$ with $X_n \searrow 0$ as $n\to \infty$ and all $\varepsilon >0$,
$$
\lim_{n\to \infty} \DR\left( X_n- \varepsilon \right)=0.
$$
\item For every sequence $(A_n)_{n\in \N}\subset \FF$ with $A_{n+1}\subset A_n$ for all $n\in \N$ and $\bigcap_{n\in \N} A_n=\emptyset$,
$$
\lim_{n\to \infty} \DR\left( \eins_{A_n}\right)=0.
$$
\item For every sequence $(X_n)_{n\in \N}\subset \Bb$ with $X_n \searrow 0$ as $n\to \infty$,
$$
\lim_{n\to \infty} \int X_n\, \d c=0.
$$
 \end{enumerate}
 If $\DR$ is additionally submodular, either of these conditions implies that there exists a nonempty set $\cP$ of probability measures on $(\Omega,\FF)$ with
 \[
 \DR(X)=\max_{\Q\in \cP}\Q(X>0)\quad \text{for all }X\in \Bb.
 \] 
\end{proposition}

\begin{example}
Let $\P$ be a probability measure on $\FF$ and $T\colon [0,1]\to [0,1]$ be nondecreasing with $T(0)=0$ and $T(1)=1$.\ Then, the default risk measure $\DR\colon \Bb\to [0,1]$, given by
\[
\DR(X):=T\big(\PD_\P(X)\big)\quad\text{for all }X\in \Bb,
\]
satisfies Property (ii) in Proposition \ref{prop.contabove} if $\inf_{p>0}T(p)=0$.\ In case there exists a sequence $(A_n)_{n\in \N}\subset \FF$ with \JS{$\emptyset \neq A_{n+1}\subset A_n$ for all $n\in \N$} and $\bigcap_{n\in \N} A_n=\emptyset$, this is also a necessary condition for $\DR$ to satisfy Property (ii) in Proposition \ref{prop.contabove} \JS{using the continuity from above of $\P$}.
\end{example}

\section{Law-invariant default risk measures and distorted PDs}\label{sec:distorted}

 Let $C$ be a set of customers $C\subset \Bb$, which contains the set of all constant functions, and satisfies $\eins_{\{X>0\}}\in C$ for all $X\in C$.\ In this section, we fix a reference probability measure $\P$ on $\FF$, and specialize on \textit{law-invariant} default risk measures $\DR\colon C\to [0,1]$, i.e., $\DR(X)=\DR(Y)$ whenever $X\in C$ and $Y\in C$ have the same distribution under $\P$. In rating systems, customers in the same rating class are considered to be statistically identical in terms of their default behaviour.\ Hence, when choosing a default risk measure in order to quantify the probability of default including uncertainty, it makes sense to require law-invariance. Since customers are usually divided into rating classes and not every default probability is realized, we also consider the case where $(\Omega,\FF,\P)$ is \textit{not} atomless. 
 We start with several characterizations of law-invariance, which do not hinge on the standard assumption of an atomless probability space, before we switch to an atomless setting in order to derive finer properties and representations of law-invariant default risk measures.
 
 In the sequel, we say that a function $T\colon [0,1]\to [0,1]$ is a \textit{distortion function} if $T(0)=0$ and $T(1)=1$. The following theorem adopts an argument from Wang et al.\ \cite[Proof of Theorem 2]{MR1604936}, and provides a characterization of distorted PDs.

\begin{theorem} \label{Theorem: law inv, Dist}
 Let $\DR\colon C\to [0,1]$ be a default risk measure.\ Then, the following statements are equivalent.
\begin{enumerate}
\item[(i)] $\DR$ is law-invariant and $\DR(X)=\DR\big(\eins_{\{X>0\}}\big)$ for all $X\in C$.
\item[(ii)] There exists a distortion function $T\colon [0,1]\to [0,1]$ with 
\[
\DR(X)=T\big(\PD_\P(X)\big)\quad\text{for all }X\in C.
\]
\end{enumerate}
\end{theorem}

A natural question that arises is whether the distortion function $T$ in Theorem \ref{Theorem: law inv, Dist} is nondecreasing or, in other words, if the default risk measure $\DR\colon C\to [0,1]$ is \textit{consistent} with $\P$, i.e., 
\[
\DR(X)\leq \DR(Y) \quad \text{for all }X,Y\in C \text{ with } \PD_\P(X)\leq \PD_\P(Y). 
\]
This property is very natural since one would expect the default risk of $X$ to be smaller than the default risk of $Y$ if the PD of $X$ is smaller than the PD of $Y$. In the situation of Theorem \ref{Theorem: law inv, Dist}, it is, however, possible that the distortion function $T$ is \textit{not} monotone as the following simple example shows.

\begin{example} \label{ex. decreasing dist}
Let $\Omega=\{0,1\}$, $\FF$ be the power set, and $C$ consist of all constants and the two functions $X:=\eins_{\{0\}}$ and $Y:=\eins_{\{1\}}$. Assume that
\[
0<p:=\P(X>0)<\P(Y>0)=:q<1\text{ and } 0<\DR(Y)<\DR(X)<1.
\]
Then $\DR$ is a default risk measure, which satisfies Property (i) in Theorem \ref{Theorem: law inv, Dist} but, for any distortion function $T\colon [0,1]\to [0,1]$ with
\[
\DR(Z)=T\big(\PD_\P(Z)\big)\quad\text{for all }Z\in C,
\]
it follows that $T(p)=\DR(X)>\DR(Y)=T(q)$, so that $T$ is not monotone. 
\end{example}

We now aim towards a characterization in terms of a nondecreasing distortion function. The proof of Theorem \ref{thm.extension} indicates that the set
\begin{equation}\label{eq.defsetP}
P:=\big\{p\in [0,1]\, \big|\, \exists X\in C\colon \P(X>0)=p\big\}. 
\end{equation}
plays a fundamental role for the monotonicity of the distortion function.
\begin{definition}
We say $C$ contains an \textit{ordered subset} if there exists a family $\big( X_p \big)_{p\in P}$ with $\PD_\P(X_p)=p$ for all $p\in P$ and $\{X_{p} >0 \} \subset \{X_{q} >0 \}$ for all $p,q\in P$ with $p\leq q$.
\end{definition}
Clearly, if $(\Omega,\FF,\P)$ is atomless and $C=\Bb$, then $C$ contains an ordered subset.

\begin{theorem} \label{thm. eq dist. PD}
 Let $\DR\colon C\to [0,1]$ be a default risk measure, and assume that $C$ contains an ordered subset.\ Then, the following statements are equivalent.
\begin{enumerate}
\item[(i)] $\DR$ is law-invariant and $\DR(X)=\DR\big(\eins_{\{X>0\}}\big)$ for all $X\in C$.
\item[(ii)] There exists a nondecreasing distortion function $T\colon [0,1]\to [0,1]$ with
\[
\DR(X)=T\big(\PD_\P(X)\big)\quad\text{for all }X\in C.
\]
\end{enumerate}
\end{theorem}

\begin{remark}
Under the assumptions of Theorem \ref{thm. eq dist. PD}, including one of the equivalences, the default risk measure $\DR$ can be extended to a law-invariant default risk measure $\overline \DR\colon \Bb\to [0,1]$ by means of the distortion function $T$. The extension $\overline \DR$ is given by
$$\overline \DR(X):=T\big(\PD_\P(X)\big)\quad\text{for all }X\in \Bb.$$ Obviously, $\overline \DR(X)=\DR(X)$ for all $X\in C$ and $\overline \DR$ is a default risk measure, cf.\ Example \ref{Distorted PD}.\ Note that a default risk measure can, in general, not be extended using a nonmonotone distortion function as in Example \ref{ex. decreasing dist}.\ This can be seen by considering, for example, the set $\Omega=\{0,1,2\}$ together with the power set $\FF$ and $C$ consisting of $X:=\eins_{\{0\}}$, $Y:=\eins_{\{1\}}$ and all constant functions.\ Let $\P(X>0)=0.5$, $\P(Y>0)=0.3$, $\DR(X)=0.5$, $\DR(Y)=0.7$, and $T\colon [0,1]\to [0,1]$ be a distortion function with
\[
\DR(Z)=T\big(\PD_\P(Z)\big)\quad\text{for all }Z\in C.
\]
Then, for $U:=\eins_{\{1,2\}}$, it follows that $\P(U>0)=0.5=\P(X>0)$. Hence, $$T\big(\PD_\P(U)\big)=T(0.5)=\DR(X)=0.5<0.7=\DR(Y)$$
despite the fact that $Y\leq U$.
\end{remark}

\begin{remark}\label{rem.facts.distorted}
We recall some well-known facts about distorted probabilities, cf.\ \cite{MR3859905}. For the reader's convenience, we provide short proofs of some of the statements collected in this remark in the Appendix \ref{app.distorted}. In the following, let $T\colon [0,1]\to [0,1]$ be a distortion function.
\begin{enumerate}
\item[a)] A well-known fact is that the capacity $c\colon \FF\to [0,1]$, given by
\begin{equation}\label{eq.distcapa}
c(A):=T\big(\P(A)\big)\quad\text{for all }A\in \FF.
\end{equation}
is 2-alternating if $T$ is concave. If $(\Omega,\FF,\P)$ is atomless, also the converse statement holds, cf. \cite[Section 4.6]{MR3859905}.
\item[b)] A capacity $c\colon \FF\to [0,1]$ is called \textit{exact} if there exists a set of countably additive probability measures $\cP$ on $\FF$ with
\[
c(A)=\sup_{\Q\in \cP} \Q(A)\quad\text{for all }A\in \FF.
\]
It is well-known that the distorted probability $c\colon \FF\to [0,1]$, given by \eqref{eq.distcapa}, is an exact capacity if, for all $p,q\in (0,1)$ with $p<q$,
\begin{equation}\label{eq.starshaped}
\frac{T(p)}{p}\geq \frac{T(q)-T(p)}{q-p}\geq \frac{1-T(q)}{1-q}.
\end{equation}
If $(\Omega,\FF,\P)$ is atomless, also the converse holds true, cf. \cite{MR2442200}. In this case,
\begin{equation}\label{eq.rep.distorted}
T(p)=\sup_{\Q\in \cP}\int_{1-p}^1q_\Q(s)\, \d s\quad\text{for all }p\in [0,1],
\end{equation}
where $q_\Q$ denotes the quantile function of the density $\frac{\d \Q}{\d \P}$, see, e.g., \cite[Lemma 4.60]{MR3859905}. In particular,
\[
\moc(p)=\frac{T(p)}{p}-1=\sup_{\Q\in \cP}\ES_\P^{1-p}\big(\tfrac{\d \Q}{\d\P}\big)-1\quad\text{for all }p\in (0,1],
\]
where, for $\alpha\in (0,1]$ and $\Q\in \cP$, $\ES_\P^\alpha\big(\frac{\d \Q}{\d \P}\big):=\frac{1}{1-\alpha}\int_{\alpha}^1q_\Q(s)\, \d s$ denotes the \textit{expected shortfall} of the density $\frac{\d \Q}{\d \P}$ with confidence level $\alpha$. Note that \eqref{eq.starshaped} is, for example, satisfied if $T$ is concave, and observe that \eqref{eq.starshaped} implies that the function $T$ is nondecreasing and absolutely continuous as soon as
\begin{equation}\label{eq.distcont}
\inf_{p\in (0,1]}T(p)=0.
\end{equation}
Recall that a monotone function is a.e.\ differentiable and that absolute continuity of $T$ implies that $$T(p)=\int_{0}^p T'(s)\, \d s\quad\text{for all }p\in [0,1],$$ where $T'$ denotes the weak derivative of $T$.\ 
Moreover, \eqref{eq.starshaped} implies that
\[
\moc(q)\leq \moc (p)\quad\text{for all }p,q\in (0,1]\text{ with }p\leq q.
\]
In fact, assume that \eqref{eq.starshaped} is satisfied, and let $p,q\in (0,1]$. If $p=q$, the statement is trivial, and if $p< q$,
\[
\frac{T(p)\tfrac{q}{p}-T(p)}{q-p}=\frac{T(p)}{p}\geq \frac{T(q)-T(p)}{q-p},
\]
which yields that $T(p)\frac{q}{p}\geq T(q)$.\\
\end{enumerate}
\end{remark}

We conclude this section with the following characterization of the minorants of distorted PDs, which can be found in a similar yet different form in \cite[Theorem 4.79]{MR3859905}.

\begin{proposition} \label{prop. atomless}
    \JS{Let $T\colon [0,1]\to [0,1]$ be a distortion function, which satisfies \eqref{eq.starshaped} and \eqref{eq.distcont}, and $\Q$ be a probability measure on $\FF$. If
    \begin{equation}\label{eq.cond.distineq}
    \int_0^p q_\Q(1-s)\, \d s\leq T(p)\quad\text{for all }p\in [0,1],
    \end{equation}then
    \begin{equation}\label{eq.distineq}
    \PD_\Q(X)\leq T\big(\PD_\P(X)\big)\quad\text{for all }X\in \Bb.
    \end{equation}
    If $(\Omega,\FF,\P)$ is atomless, also the converse holds true.}
\end{proposition}

\begin{remark}
\JS{Consider the situation of Proposition \ref{prop. atomless}.\ In view of Remark \ref{rem.facts.distorted}, a sufficient condition for \eqref{eq.cond.distineq} and thus \eqref{eq.distineq} to be satisfied is that $q_\Q(1-p)\leq T'(p)$ for a.a.\ $p\in (0,1)$.\ That this is, however, not a necessary condition, can easily be seen by considering $T(p)=\sqrt{p}$ for all $p\in [0,1]$ and $\Q=\P$. Then, $q_\P(1-p)=1$ for all $p\in (0,1)$ and $T'(p)=\frac1{2\sqrt{p}}<1$ for all $p\in \big(\frac14,1\big)$. However, $$T(p)=\sqrt{p}\geq p=\int_0^p q_\P(1-s)\, \d s.$$}
\end{remark}

\section{A Case Study on Capital Requirements for Financial Institutions}\label{sec:application}
In this section, we link our axiomatic study of default risk measures to financial institutions' capital requirements accounting for model uncertainty.\ According to the guidelines of the European Banking Authority \cite{EBA_GL}, the PDs of a rating system are calibrated to a `best estimate' level without using systematically conservative input values for the calculation.\ Those PDs are then used in the risk-oriented group management, among others.\ In order to establish the connection between model uncertainty and capital requirements, we are guided by Article 179 (f) and Article 180 (e) of the CRR \cite{CRR} that determine that an appropriate margin of conservatism, reflecting the expected range of estimation errors, must be formed for the `best estimate' PD of the rating system.

Usually, one differentiates between the \textit{expected loss} (EL) and unexpected losses, which are covered by \textit{risk-weighted assets} (RWAs).\ According to CRR Article 153 \cite{CRR}, the dependence of RWAs on the probability of default including model uncertainty is described by the function $\text{RWA}\colon [0,1]\to \R_{\geq0}$, via $\text{RWA}(0)=0$, $\text{RWA}(1)=1$, and
    \begin{equation} \label{eq:RWA}
    \text{RWA}(p):=1.06\cdot12.5\cdot \text{EaD} \cdot  \text{LGD}\Bigg( N\bigg(\frac{G(p)-\sqrt{R(p)}\cdot G(0.999)}{\sqrt{1-R(p)}}  \bigg) -p\Bigg)
    \end{equation} 
    with
$$R(p):=0.12\cdot \frac{1-e^{-50 p}}{1-e^{-50}}+0.24\cdot \left(1-\frac{1-e^{-50 p}}{1-e^{-50}} \right)\quad\text{for all }p\in (0,1),$$\\
where $N$ is the cumulative distribution function of the standard normal distribution, $G$ denotes the inverse distribution function of the standard normal distribution, and $R$ can be seen as a correlation factor.\footnote{For simplicity, we assume the effective maturity $M$ in Article 153 of the CRR \cite{CRR} to be equal to $1$.} On a practical level, a frequent choice for the PD including model uncertainty is the `best-estimate' PD of the rating system multiplied with suitable a margin of conservatism ($1+\moc$ as a multiplier). In order to get a better understanding for the RWA formula, we briefly explain the terms EaD and LGD. The \textit{exposure at default} (EaD) can be seen as the amount of credits of a borrower at the time of its default, for instance 1 million. The \textit{loss given default} (LGD) on the other hand is the height of the loss in relation to the amount of exposure at the time of default, i.e., it is a number between $0$ and $1$.

Risk-weighted assets are essential for financial institutions' capital requirements, since they must hold at least 8\% of RWAs as equity capital.\ Hence, they play an important role for financial institutions' risk provisions. In comparison to the expected loss (EL), which is calculated as the product of `best-estimate' PD, LGD, and EaD. RWAs focus on unexpected losses from exposures, which show different characteristics compared to expected losses.  

\begin{figure}[htb]
\includegraphics[width=\textwidth]{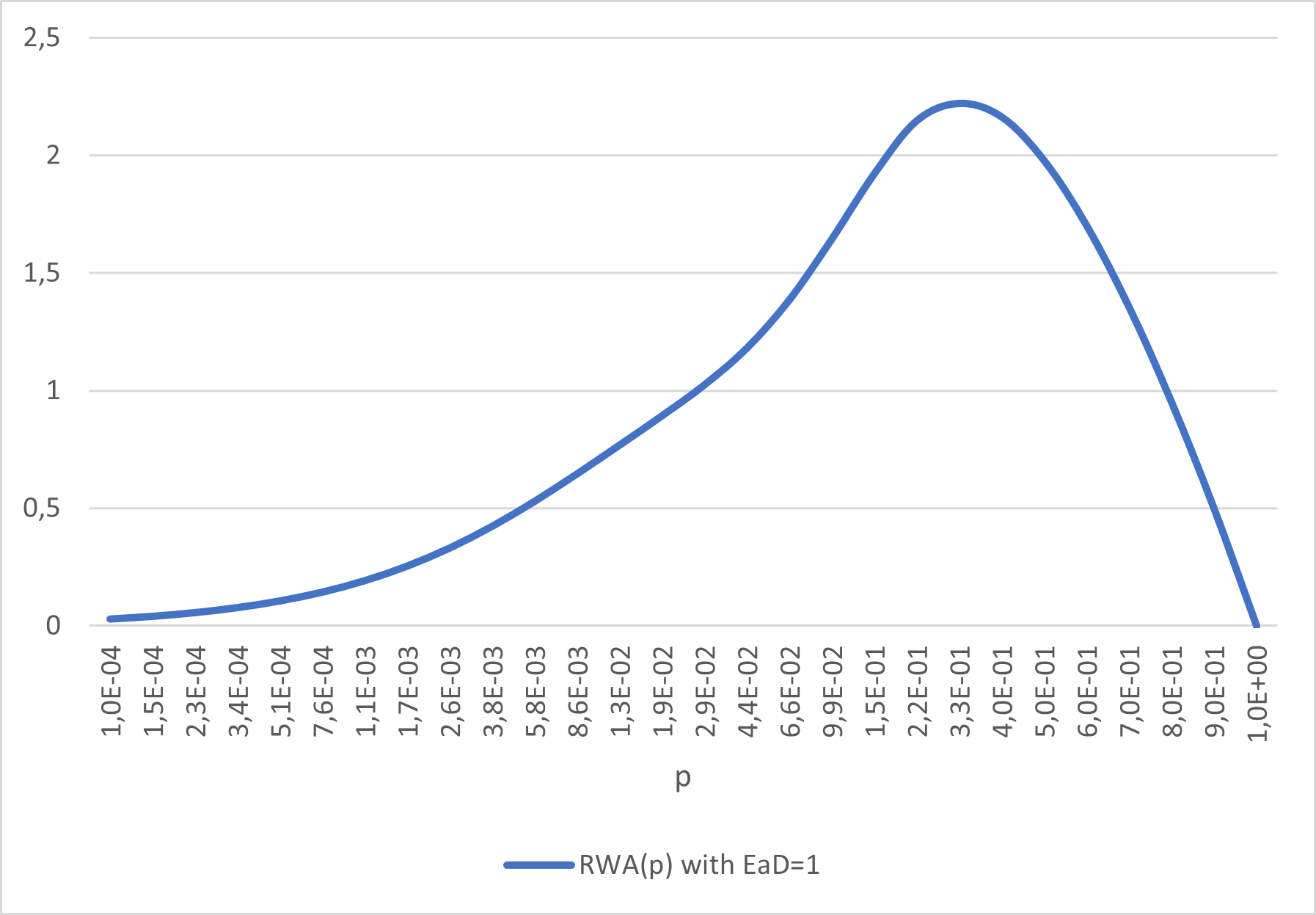}
\caption{RWA dependence on PD}
\label{RWA EAD}
\end{figure}
In Figure \ref{RWA EAD}, we see the dependence of RWAs on $p$, where we assume the EaD to be 1 and the LGD to be 0.4.\ The standardization of setting the EaD as $1$ is also referred to as risk weight (RW) in Article 153 of the CRR \cite{CRR}.\ We observe that the function is not monotonically increasing with respect to the probability of default $p\in [0,1]$ -- a characteristic that applies, e.g., to expected losses.\ This property is not surprising since for very high PDs the loss is expected and thus the unexpected loss becomes smaller (for the same exposure). Nonetheless, this could potentially lead to problems concerning risk provisions as we will see in Example \ref{last ex} below. 

In the following examples, we now focus on a concrete rating system consisting of 22 rating classes and regard two distinct methods of quantifying model uncertainties. Our aim is to analyze the effect of different default risk measures on capital requirements.

For the calculation of capital requirements corresponding to RWAs, the PD including the MoC (hence including model uncertainty) is used.\ In general, a common problem occurring in the course of a model estimation is the lack of sufficient data.\ As a consequence, the model is sometimes estimated on pooled data, using combined information of many financial institutions in order to attain a higher validity.

 There are of course several ways to specify a margin of conservatism.\ On the one hand, it is conceivable to calculate a MoC that is constant for all rating classes, since uncertainties cannot always be precisely quantified and the total data provides a more stable result for the entire rating system.\ However, on the other hand, the highly different amount of data per rating class strongly suggests that the uncertainty of the model should depend on the rating class as well.\ In the following examples, we compare the influence of these different approaches on capital requirements.

To display the 22 rating classes, we fix a discrete probability measure $\P$ and the reference default risk measure $\PD_\P (X):=\P(X>0)$. To express the variable model uncertainty, which, loosely speaking, will be referred to as the `true' model uncertainty, we use a distorted PD (which is a default risk measure according to \ref{Distorted PD}), where the distortion function $T\colon [0,1]\to [0,1]$ is assumed to be nondecreasing with $T(0)=0$ and $T(1)=1$, and the margin of conservatism is given by
$$\text{MoC}(p)=\frac{T\big(p\big)}{p}-1\quad\text{for all }p\in (0,1].$$
If we allowed $T$ to be decreasing on a subinterval of $[0,1]$, then $\DR(X):=T\big(\PD_\P(X)\big)$ is, in general, no default risk measure, cf. Example \ref{ex. decreasing dist}.

From a practical point of view, the exclusion of decreasing functions $T$ in not a big restriction, since it would imply that there are at least two neighboring rating classes $k$ and $k+1$ with corresponding PDs denoted by $p_k$ and $p_{k+1}$ such that $T(p_k)$ is larger than $T(p_{k+1}).$ Assuming, for example, that the PD of two neighboring rating classes grows by the factor $1.5$, i.e., $p_{k+1}=1.5\cdot p_k$, we would attain the inequality 
\[
\moc(p_k)=\frac{T\big(p_k\big)}{p_k}-1>\frac{T\big(p_{k+1}\big)}{p_k}-1=1.5\cdot \frac{T\big(p_{k+1}\big)}{p_{k+1}}-1=1.5\cdot \moc ( p_{k+1}) +0.5,
\]
where the uncertainty of the rating classes $k$ and $k+1$ is expressed by $\moc(p_k)$ and $\moc(p_{k+1})$, respectively.\ If, for instance, $\moc ( p_{k+1} )=0.2$, then $\moc(p_k)>0.8$, which implies a significantly larger uncertainty in rating class $k$ both in relative and absolute numbers.\ Such a difference between two rating classes neither seems to be realistic nor is observed in practice. 

\begin{example}[Monotonically decreasing uncertainty with respect to the rating class] \label{Appl. Example} In this example, we focus on a model where the `true' uncertainty decreases with increasing rating classes, i.e., for two arbitrary rating classes $k$ and $k+1$ with PDs $p_k$ and $p_{k+1}$, respectively, we have 
\[
\moc(p_k)\geq \moc(p_{k+1}).
\]
Note that this is consistent with the choice of a distortion function $T$ that satisfies \eqref{eq.starshaped}, which we will assume throughout the remainder of this example. As a consequence, the distorted PD has a representation over probability measures, which is not unusual \JS{since the lack of defaults in the `good' rating classes typically makes it very difficult to adequately estimate a PD} -- thus the uncertainty here is very high. In the \textit{intermediate} rating classes, there are both a lot of customers and defaults. The set of very `bad' customers is often smaller but most of the defaults occur there implying a higher validity.

In Figure \ref{Distorted PD 1}, we depict three different functions.\ The green function is just the identity function, the dashed blue line is the model uncertainty where a constant MoC is used for all rating classes, and the orange line describes the uncertainty through the function $T(p)=p+\frac{\sqrt{p}-p}{40}$ for $p\in [0,1]$, which represents a `small' perturbation of the identity expressing the `true' uncertainty of the model per rating class.\ We observe that
\begin{figure}[htb]
\includegraphics[width=\textwidth]{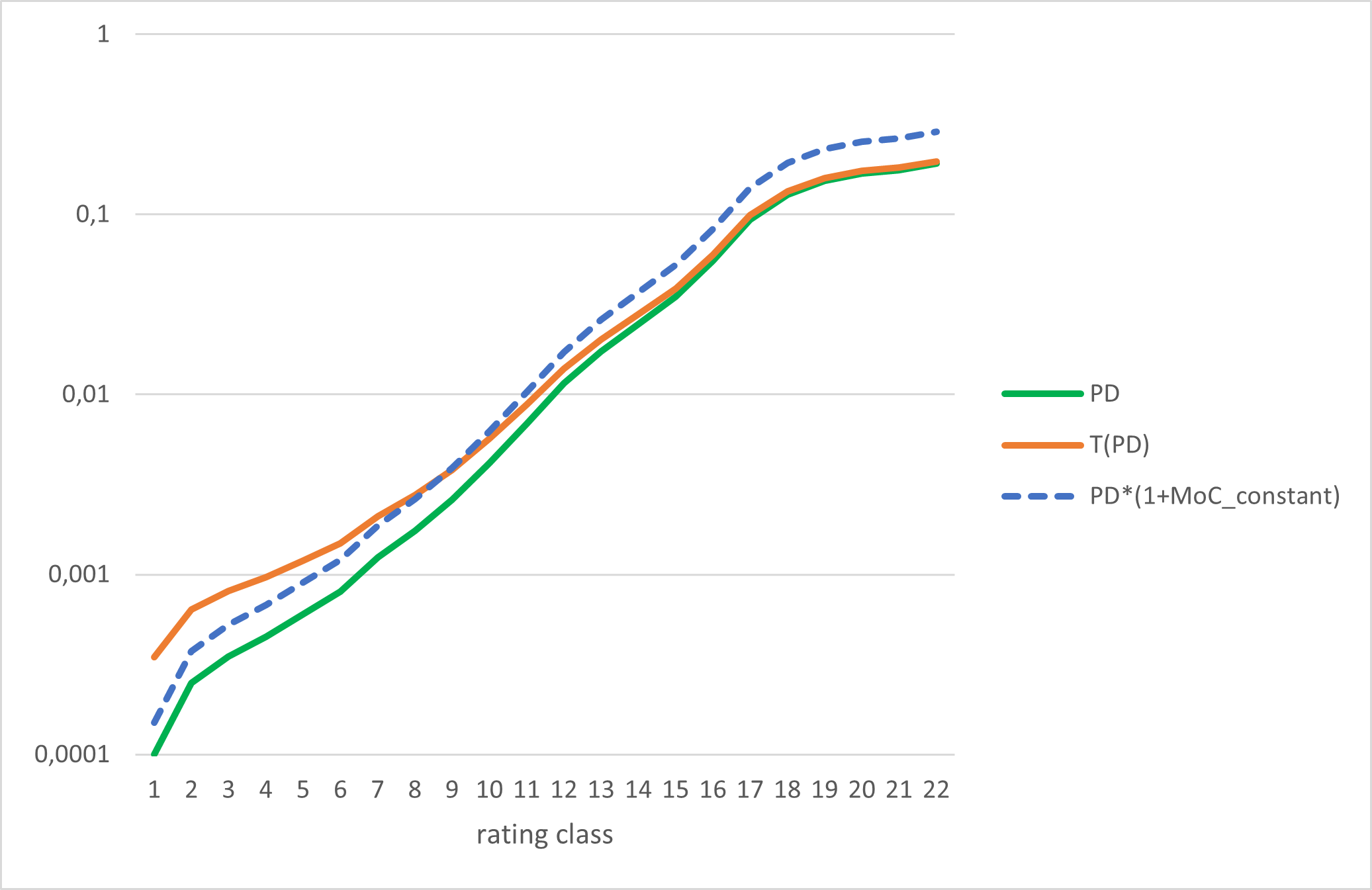}
\caption{PDs per rating class including and excluding uncertainty}
\label{Distorted PD 1}
\end{figure}
$$T(p_k)> p_k\big(1+\text{MoC}_{\text{constant}}\big)\text{ for }k\leq8$$
and
$$T(p_k)< p_k\big(1+\text{MoC}_{\text{constant}}\big)\text{ for }k>8.$$
To illustrate the impact of different uncertainties on RWAs, we consider the quotient between RWAs with model uncertainty and RWAs without model uncertainty once with variable uncertainty and then with a constant uncertainty per rating class. In other words, the quotient tells us with which factor the capital requirement grows per rating class due to uncertainty.

Unsurprisingly, the described factor in Figure \ref{RWA 1} is larger for the variable uncertainty in the rating classes $1$ to $8$. For the other rating classes, it is the other way around. In this case, we can conclude that a higher model uncertainty leads to more RWAs and thus a higher amount of capital requirement.\ As a result, we could have a scenario where a financial institution uses the model uncertainty with a constant MoC to calculate its risk provisions. If the institution has customers primary in the good rating classes then it would not have enough capital requirement, again assuming that the `true' model uncertainty is expressed by $T$.\ On the other hand, an institution with customers in bad rating classes would have a larger amount of capital requirement \JS{than 
necessary by regulation}, thus leading to a possible competitive disadvantage.
\begin{figure}[htb]
\includegraphics[width=\textwidth]{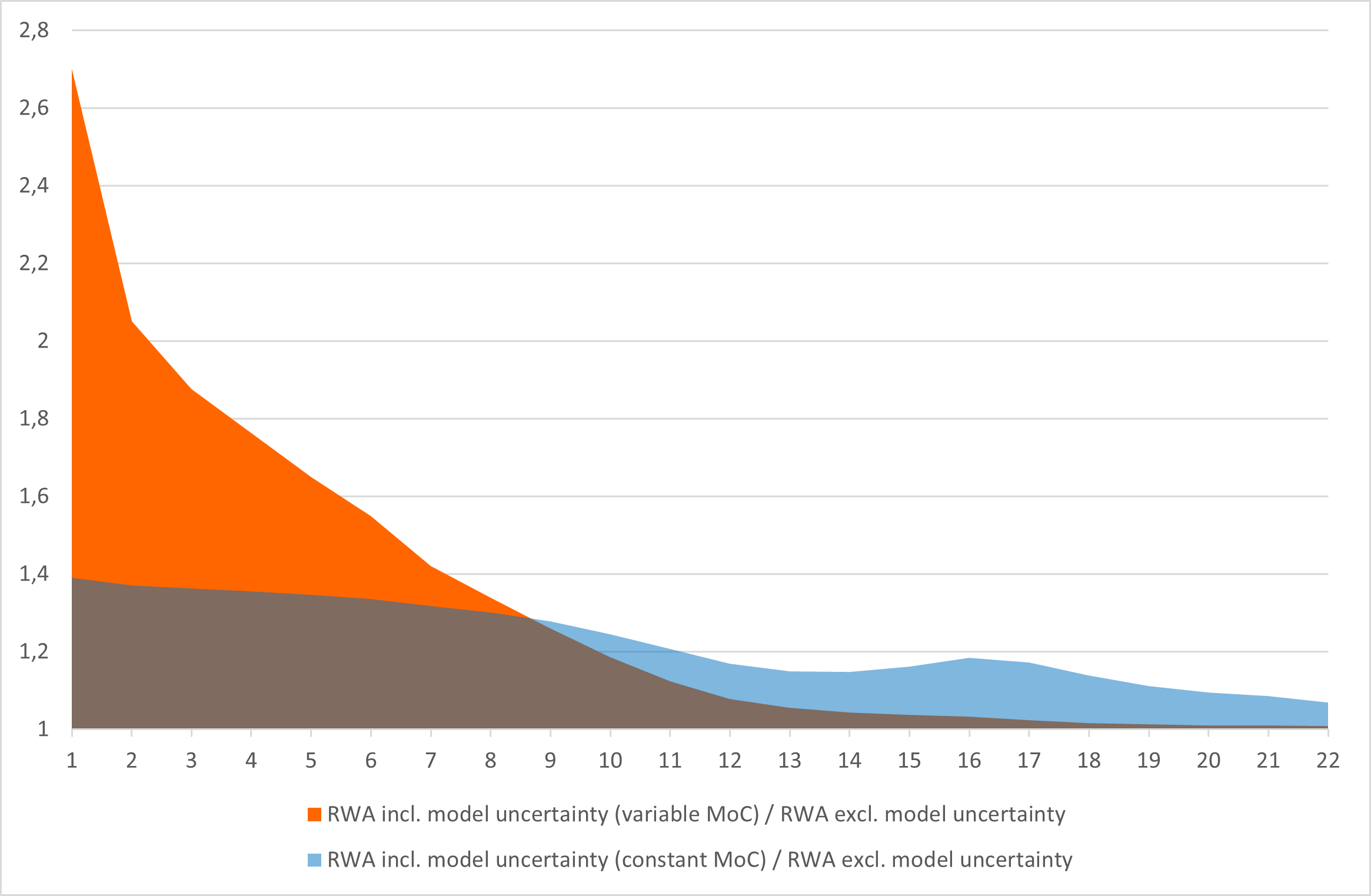}
\caption{RWA growth due to uncertainty}
\label{RWA 1}
\end{figure}
\end{example}

\begin{example}[Nonmonotonic uncertainty per rating class] \label{last ex} In a second example, we focus on a rating model that is estimated only on customers that are in good or middle rating classes. Thus, the uncertainty for bad rating classes is high since the calibration function has to be extrapolated on it without direct information from backtesting or benchmarking data.\ Such constellations can appear in low default portfolios where, historically, no `bad' customers have been observed yet. Nonetheless, it could happen that some customers will be classified in bad rating classes in the future and, in that case, the accuracy of the extrapolation and its uncertainty cannot be neglected.

Again, the distortion function $T$ expresses the `true' degree of uncertainty, and is assumed to be nondecreasing. Hence, the `true' uncertainty of the model again depends on the rating classes and peaks on both the very good and the very bad classes. The distortion function $T$ can explicitly be constructed in such a way that the orange line lies above the dashed blue line for the rating classes $1$ to $4$ and $18$ to $22$.

Eye-catching is the fact that, for the rating classes from $20$ to $22$, the growth of RWAs caused by uncertainty is smaller when, instead of the constant MoC, the distortion function $T$ is used although the uncertainty here is much higher.\ This can be explained by recalling the property of RWAs of not being monotonically increasing with respect to the PD (including uncertainty) since the unexpected loss gets smaller for very high PDs.\ However, regarding capital requirements, model uncertainty is only considered for RWAs and not for ELs.\ So, in contrast to the regulatory intention to increase equity capital by considering model uncertainty, for those rating classes ($20$ to $22$) the capital requirement is de facto smaller when higher model uncertainty is taken as a basis. 
\begin{figure}[htb]
\includegraphics[width=\textwidth]{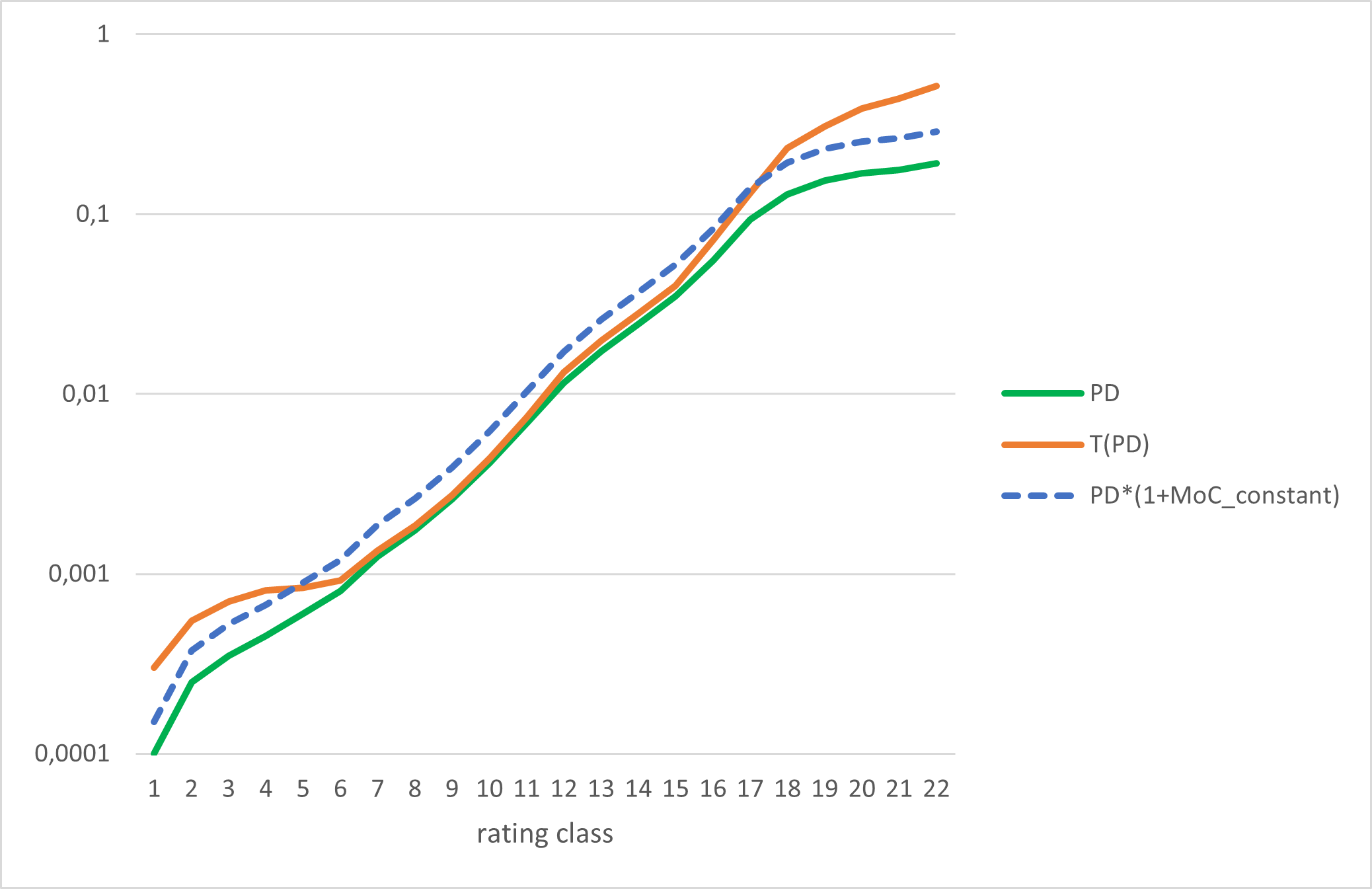}
\caption{PDs per rating class including and excluding uncertainty}
\end{figure}
\begin{figure}[htb]
\includegraphics[width=\textwidth]{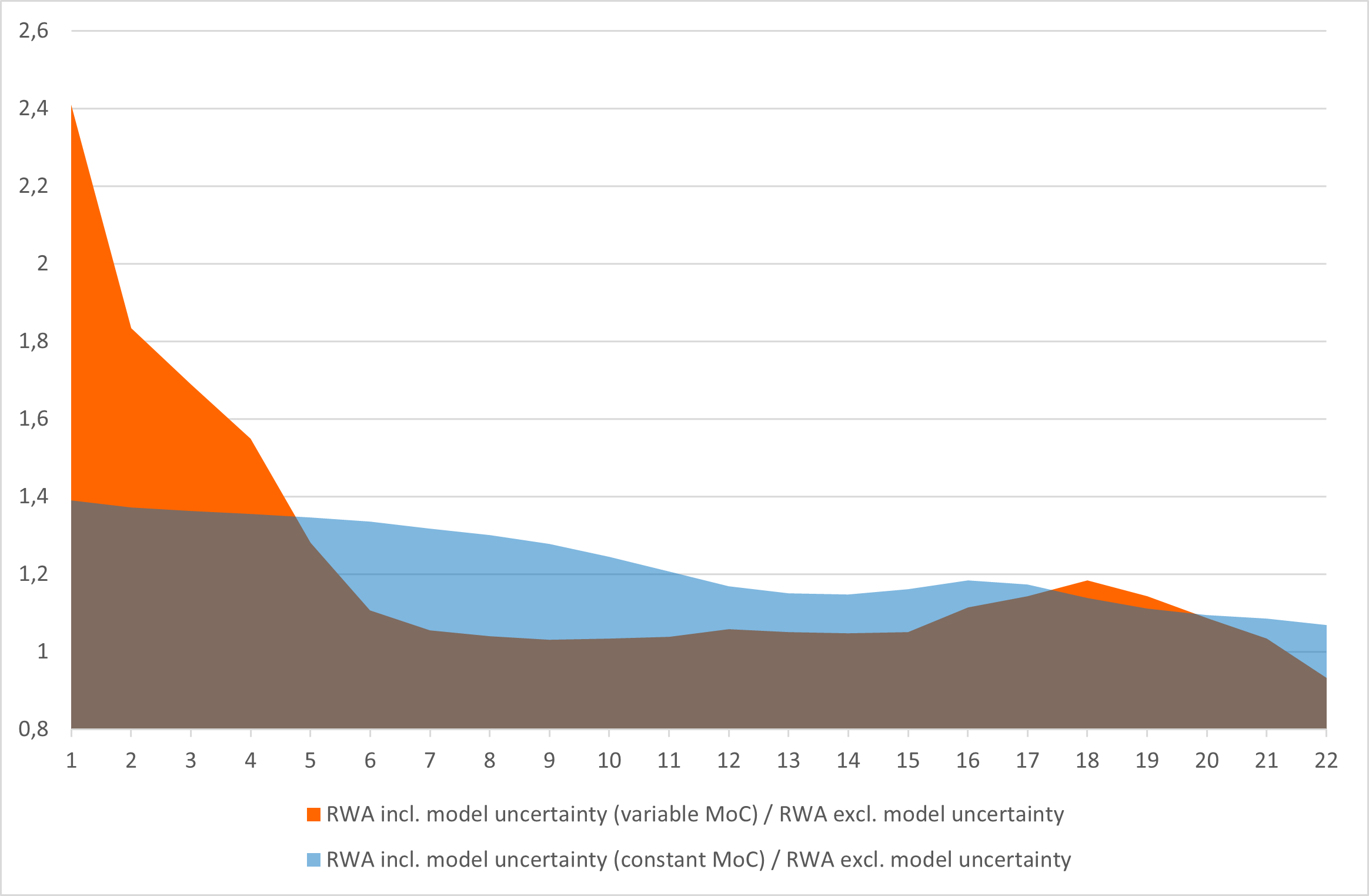}
\caption{RWA growth due to uncertainty}
\end{figure}
\end{example}

\appendix

\section{Distorted probabilities}\label{app.distorted}

Throughout this section, let $(\Omega, \mathcal F,\P)$ be a probability space and $T\colon [0,1]\to [0,1]$ with $T(0)=0$ and $T(1)=1$. Moreover, let
\[
c(A):=T\big(\P(A)\big)\quad\text{for all }A\in \FF.
\]
The following lemma can be found, e.g., in \cite{MR3859905}.\ For the reader's convenience, we provide a short proof.

\begin{lemma}\label{lem.coheren}
The capacity $c$ is $2$-alternating if $T$ is concave.\ If $(\Omega,\FF,\P)$ is atomless, also the converse is true.
\end{lemma}

\begin{proof}
First assume that $T$ is concave, and let $A,B\in \FF$. Since $\P(A\cap B)\leq \P(A)\leq \P(A\cup B)$, there exists some $\lambda\in [0,1]$ such that
\[
\P(A)=\lambda\P(A\cup B)+(1-\lambda)\P(A\cap B).
\]
Since $\P(A)+\P(B)=\P(A\cup B)+\P(A\cap B)$, it follows that
\[
\P(B)=(1-\lambda)\P(A\cup B)+\lambda\P(A\cap B).
\]
Hence, the concavity of $T$ implies that
\[
T\big(\P(A)\big)+T\big(\P(B)\big)\geq T\big(\P(A\cup B)\big)+T\big(\P(A\cap B)\big).
\]
Now, assume that $(\Omega,\FF,\P)$ is atomless and that $c$ is $2$-alternating.\ By Remark \ref{rem.distpdcontbelow} b), there exists a set $\cP$ of countably additive and (w.r.t.\ $\P$) absolutely continuous probability measures with
\[
\int X\,\d c=\sup_{\P\in \cP} \E_\P(X)\quad\text{for all }X\in \Bb.
\]
 Let $p,q\in [0,1]$ with $p\leq q$.\ Since $(\Omega,\FF,\P)$ is atomless, there exist $A_1,A_2\in \FF$ with $\P(A_1)=p$, $\P(A_2)=q$, and $A_1\subset A_2$. Let $\varepsilon>0$ and $X:=\eins_{A_1}+\eins_{A_2}$.\ Then, there exists some $\Q\in \cP$ with
 \[
 c(A_1)+c(A_2)=\int X\, \d c\leq \E_\Q(X)+\varepsilon=\Q(A_1)+\Q(A_2)+\varepsilon.
 \]
  Since $\Q(A_1)\leq c(A_1)$ and $\Q(A_2)\leq c(A_2)$, it follows that $c(A_1)\leq \Q(A_1)+\varepsilon$ and $c(A_2)\leq\Q(A_2)+\varepsilon$.\ By \cite[Lemma 4.60]{MR3859905}, for all $\lambda\in [0,1]$,
 \begin{align*}
T\big((1-\lambda)p+\lambda q\big)&\geq \int_{1-(1-\lambda)p-\lambda q}^1q_\Q(s)\, \d s\geq 
(1-\lambda)\int_{1-p}^1q_\Q(s)\, \d s+\lambda \int_{1-q}^1q_\Q(s)\, \d s\\
&\geq (1-\lambda )c(A_1)+\lambda c(A_2)-\varepsilon=(1-\lambda )T(p)+\lambda T(q)-\varepsilon,
 \end{align*}
 where, in the second step, we used the fact that the map $[0,1]\to \R, \; t\mapsto \int_{1-t}^1q_\Q(s)\, \d s$ is concave.
\end{proof}

\begin{remark}\label{rem.starshaped}
By \cite[Lemma 3.1]{MR2442200}, the function $T$ satisfies \eqref{eq.starshaped} if and only if \begin{equation}\label{eq.supT}T(p)=\sup_{i\in I} T_i(p)\quad \text{for all }p\in [0,1],\end{equation} where $I$ is a nonempty set and $T_i\colon [0,1]\to [0,1]$ is concave with $T_i(0)=0$ and $T_i(1)=1$ for all $i\in I$.\ In this case, the supremum in \eqref{eq.supT} is in fact a maximum.
\end{remark}

Using Lemma \ref{lem.coheren}, we can prove the following result from \cite{MR2442200}, see also \cite{MR1401848}. 

\begin{lemma}\label{lem.starshaped}
The capacity $c$ is exact if $T$ satisfies \eqref{eq.starshaped}.\ If $(\Omega,\FF,\P)$ is atomless, also the converse is true.
\end{lemma}

\begin{proof}
First assume that $T$ satisfies \eqref{eq.starshaped}. Then, by Remark \ref{rem.starshaped}, $T(p)=\sup_{i\in I} T_i(p)$ for all $p\in [0,1]$, where $I$ is a nonempty set and $T_i\colon [0,1]\to [0,1]$ is concave with $T_i(0)=0$ and $T_i(1)=1$ for all $i\in I$. By Lemma \ref{lem.coheren}, for each $i\in I$, there exists a nonempty set $\cP_i$ of countably additive probability measures with
\[
T_i\big(\P(A)\big)=\sup_{\Q\in \cP_i} \Q(A)\quad\text{for all }A\in \FF.
\]
Let $\cP:=\bigcup_{i\in I}\cP_i$. Then,
\[
T\big(\P(A)\big)=\sup_{i\in I}T_i\big(\P(A)\big)=\sup_{i\in I}\sup_{\Q\in \cP_i} \Q(A)=\sup_{\Q\in \cP} \Q(A)\quad\text{for all }A\in \FF.
\]
Now, assume that $(\Omega,\FF,\P)$ is atomless and that $c$ is exact.\ Let $\cP$ be a nonempty set of countably additive probability measures with 
\[
T\big(\P(A)\big)=\sup_{\Q\in \cP} \Q(A)\quad\text{for all }A\in \FF.
\]
Then, By Remark \ref{rem.distpdcontbelow} and \cite[Lemma 4.60]{MR3859905},
\[
T(p)=\sup_{\Q\in \cP}\int_{1-p}^1q_\Q(s)\, \d s\quad\text{for all }p\in [0,1].
\]
Since the map $[0,1]\to \R, \; t\mapsto \int_{1-t}^1q_\Q(s)\, \d s$ is concave for all $\Q\in \cP$, it follows that $T$ satisfies \eqref{eq.starshaped} by Remark \ref{rem.starshaped}.
\end{proof}

\begin{remark}
	Alternatively, Lemma \ref{lem.coheren} and Lemma \ref{lem.starshaped} can also be proved without invoking the Kusuoka representation of law-invariant risk measures.\ In this case, one uses the fact that, for every probability measure $\Q$, which is absolutely continuous w.r.t.\ $\P$, every set $C\in \FF$ with $\P(C)>0$, and every $\lambda\in (0,1)$, there exists a set $A\in \FF$ with $A\subset C$, $\P(A)=\lambda \P(C)$, and $\Q(A)\geq \lambda \Q(C)$. In fact,  let 
	\[
	t:=\inf\bigg\{s>0\,\bigg|\,\P\Big(\Big\{\tfrac{\d \Q}{\d \P}>  s\Big\}\cap C\Big)\leq \lambda \P(C)\bigg\}.
	\]
	Since $(\Omega,\FF,\P)$ is atomless, there exists a set $B\in \FF$ with $B\subset \big\{\tfrac{\d \Q}{\d \P}= t\big\}$ and $$\P\bigg(\Big(B\cup\Big\{\tfrac{\d \Q}{\d \P}>  t\Big\}\Big)\cap C\bigg)=\lambda \P(C)>0.$$ Then, for $A:=\big(B\cup\big\{\tfrac{\d \Q}{\d \P}>  t\big\}\big)\cap C$,
	\begin{align*}
		\Q(C)&\leq t\P(C\setminus A)+\Q(A)\leq \frac{\P(C\setminus A)}{\P(A)}\E_\P\Big( \eins_{A} \tfrac{\d \Q}{\d \P}\Big)+\Q(A)\\
		&\leq \frac{\P(C)}{\P(A)}\Q(A)=\frac{\Q(A)}{\lambda}.
		\end{align*}
	Note that $\P(C\setminus A)=(1-\lambda )\P(C)$ and $\Q(C\setminus A)\leq (1-\lambda)\Q(C)$.
\end{remark}

\section{Proofs of Section \ref{sec: Definition}} \label{proofs of sc.2}

\begin{proof}[Proof of Theorem \ref{thm.extension}]
We first show that $\DR_F(X)=\DR(X)$ for all $X\in C$. In order to do so, let $X\in C$. Since $F(X-X)=F(0)=0$, it follows that $\DR_F(X)\leq \DR(X)$. On the other hand, $\DR(X)\leq \DR(X_0)$ for all $X_0\in C$ with $F(X-X_0)\leq 0$, and we obtain that $\DR(X)\leq \DR_F(X)$. Since $C$ contains all constants, we have already verified Property (ii) in Definition \ref{Def DRM} for $\DR_F$. In order to prove Property (i), let $X,Y\in \Bb$ with $X\leq Y$. Since $F$ is monotone,
$$\big\{X_0\in C \, \big| \, F\left(X-X_0 \right) \leq 0  \big\}\subset \big\{X_0\in C \, \big| \, F\left(Y-X_0 \right) \leq 0 \big\},$$ 
which implies that $\DR_F(X)\leq \DR_F(Y)$.
\end{proof}

\begin{proof}[Proof of Corollary \ref{cor. est}]
Let $X\in \Bb$. Since $$\overline\DR(X)\leq \overline\DR(X_0)=\DR(X_0)\quad\text{for all }X_0\in C\text{ with }X\leq X_0,$$ it follows $\overline\DR(X)\leq \DR_{\sup}(X)$.
\end{proof}

\section{Proofs of Section \ref{sec.acceptance}}\label{proofs of sc.3}

\begin{proof}[Proof of Proposition \ref{prop.VR}]
In order to prove part a), let $\alpha \in (0,1)$. Then,
$$
\DR(X-m)\leq \DR(Y-m)\quad\text{for all }m\in \R\text{ and }X,Y\in \Bb\text{ with }X\leq Y.
$$
Hence, $\VR(X)\leq \VR(Y)$. Since $\DR(0)=0$, it follows that $\VR(0)=0$. Moreover, by definition of $\VR$, $\VR(X+m)=\VR(X)+m$ for all $X\in \Bb$ and $m\in \R$. We have therefore shown that $\VR$ is a monetary risk measure.

Next, we prove part b).\ To that end, let $X\in \Bb$. First assume that $\DR(X)=1$. Then, $\VR(X)>0$ for all $\alpha\in (0,1)$, so that
\[
\inf \Big(\big\{\alpha\in (0,1) \,\big|\, \VR(X)\leq 0  \big\}\cup\{1\}\Big)=1=\DR(X).
\]
Next, assume that $\DR(X)<1$. Then, for all $\alpha\in (0,1)$, $\VR(X)\leq 0$ if and only if $\DR(X)\leq \alpha$. Hence,
\[
\inf \Big(\big\{\alpha\in (0,1) \,\big|\, \VR(X)\leq 0  \big\}\cup\{1\}\Big)=\DR(X).
\]
\JS{We proceed with the proof of part c).\ If $\VR$ is positively homogeneous for all $\alpha\in (0,1)$, then $\DR$ is scaling invariant by part b).\ On the other hand, if $\DR$ is scaling invariant, then, for all $\alpha\in (0,1)$, $X\in \Bb$, and $\la>0$,
$$
\big\{ m\in \R\,\big|\,\DR(\lambda X-m)\leq\alpha\big\}= \big\{ \lambda m\in \R\,\big|\,\DR(X-m)\leq\alpha\big\}.
$$
Hence, for all $\alpha\in (0,1)$, $X\in \Bb$, and $\lambda>0$, it follows that $$\VR(\lambda X)=\lambda \VR(X).$$
It remains to prove part d). If $\DR$ is quasi-convex, it follows that the level set $$\mathcal A_\alpha:=\{X\in \Bb\, |\, \DR(X)\leq \alpha\}$$ is convex for all $\alpha\in (0,1)$.\ Hence, by \cite[Proposition 4.7]{MR3859905}, $$\VR(X):=\inf \{m\in \R\, |\, X-m\in \mathcal A_\alpha\}, \quad\text{for }X\in \Bb,$$ defines a convex risk measure for all $\alpha\in (0,1)$.} Next, assume that $\VR$ is convex for all $\alpha\in (0,1)$. Then, for all $\alpha\in (0,1)$, $\lambda\in (0,1)$, and $X,Y\in \Bb$ with $\DR(X)\leq \alpha$ and $\DR(Y)\leq \alpha$, it follows that
 \[
 \VR\big(\lambda X+(1-\lambda) Y\big)\leq \lambda \VR(X)+(1-\lambda)\VR(Y)\leq 0,
 \]
 so that, by part b), $$\DR\big(\lambda X+(1-\lambda )Y\big)\leq \alpha.$$
 This shows that $\DR$ is quasi-convex.
\end{proof}

\begin{proof}[Proof of Theorem \ref{thm.recalculate}]
Let $X,Y\in \Bb$ with $X\leq Y$. Since $R^\alpha(X)\leq R^\alpha(Y)$ for all $\alpha\in (0,1)$, it follows that $\DR(X)\leq \DR(Y)$. Moreover,  $R^\alpha(0)=0$ for all $\alpha\in (0,1)$, so that $\DR(0)=0$. On the other hand, for all $m\in \R$ with $m>0$, $R^\alpha(m)=m>0$ for all $\alpha\in (0,1)$. Therefore,
\[
\big\{\alpha\in (0,1) \,\big|\, R^\alpha(X)\leq 0 \big\}\cup\{1\}=\{1\}
\]
and it follows that $\DR(X)=1$. We have therefore shown that $\DR$ is a default risk measure.

It remains to show the equality $R^\alpha=\VR$ for all $\alpha\in (0,1)$.\ To that end, let $\alpha\in (0,1)$ and $X\in \Bb$. First, observe that $$\DR\big(X-R^\alpha(X)\big)\leq \alpha.$$
Hence, $\VR(X)\leq R^\alpha(X)$.\ Now, let $\beta\in (\alpha,1)$ and $m\in \R$ with $\DR(X-m)\leq \alpha$. Then, $R^\beta(X-m)\leq 0$, which implies that $R^\beta(X)\leq m$. Taking the supremum over all $\beta\in (\alpha,1)$ and the infimum over all $m\in \R$ with $\DR(X-m)\leq \alpha$, it follows that $R^\alpha(X)\leq \VR(X)$.
\end{proof}

\begin{proof}[Proof of Proposition \ref{Prop. VaR}]
The implication (ii)$\,\Rightarrow\,$(i) follows from Proposition \ref{prop.VR} b), once we have shown that (i) implies (ii). To that end, let $X\in \Bb$ and $\alpha\in (0,1)$. Then, for all $m\in \R$ with
\[
\sup_{\P\in \mathcal P}\P(X>m)=\DR(X-m)\leq \alpha,
\]
it follows that $m\geq \VaR^\alpha_\P(X)$ for all $\P\in \mathcal P$. This implies that
\[
\sup_{\P\in \mathcal P}\VaR^\alpha_\P(X)\leq \VR(X).
\]
Now, let $\varepsilon>0$ and
\[
m:=\sup_{\P\in \mathcal P}\VaR^\alpha_\P(X)+\varepsilon.
\]
Then, for all $\P\in \mathcal P$, $\P(X>m)\leq \alpha$. This implies that
\[
\VR(X)\leq m=\sup_{\P\in \mathcal P}\VaR^\alpha_\P(X)+\varepsilon.
\]
Taking the limit $\varepsilon \downarrow 0$, the claim follows.
\end{proof}

\begin{proof}[Proof of Proposition \ref{Prop. robust VaR}]
Again, the implication (ii)$\,\Rightarrow\,$(i) follows from Proposition \ref{prop.VR} b), once we have shown that (i) implies (ii). To that end, let $X\in \Bb$ and $\alpha\in (0,1)$. Then, for all $m\in \R$ with
\[
T\big(\P(X>m)\big)=\DR(X-m)\leq \alpha,
\]
it follows that $\P(X>m)\leq T^{-1}(\alpha)$. Now, let $m\in \R$ with
$\P(X>m)\leq T^{-1}(\alpha)$. Since $T$ is lower semicontinuous, it follows that
\[
T\big(\P(X>m)\big)\leq \alpha.
\]
Hence,
\[
\big\{m\in \R\, \big| \DR(X-m)\leq \alpha\big\}=\big\{m\in \R\, \big| \P(X>m)\leq T^{-1}(\alpha)\big\}.
\]
Taking the infimum, both, on the left and the right-hand side, the claim follows.
\end{proof}
\section{Proofs of Section \ref{sec.properties}}
\label{proofs of sc.4}

\begin{proof}[Proof of Proposition \ref{prop.condition}]
Let $X\in \Bb$ with $X\geq 0$. In Remark \ref{rk:scaling} a), we have already seen that $\DR(X)\leq \DR(\eins_{\{X>0\}})$. On the other hand, for all $\varepsilon >0$,
 \[
\DR\big(X+\varepsilon\eins_{\{X>0\}}\big)\geq  \DR\big(\varepsilon\eins_{\{X>0\}}\big)= \DR\big(\eins_{\{X>0\}}\big)
 \]
and, using again Remark \ref{rk:scaling} a),
\[
\DR(X)\geq \DR\big(X\eins_{\{X>\varepsilon\}}\big)\geq \DR\big(\varepsilon\eins_{\{X>\varepsilon\}}\big)=\DR\big(\eins_{\{X>\varepsilon\}}\big).
\]
The proof is complete.
\end{proof}

\begin{proof}[Proof of Proposition \ref{prop.capacity_represenatation}]
We start with the implication (i)$\,\Rightarrow\,$(ii). Let $m>0$ and $X\in \Bb$ with $X\geq 0$. Then, $\DR(X)\leq\DR\big(X+m\eins_{\{X>0\}}\big)$. Moreover, for every $\lambda\in (0,1)$,
\[
\DR\big(X+m\eins_{\{X>0\}}\big)=\DR\big(\lambda X+\lambda m\eins_{\{X>0\}}\big)\leq \DR\big(X+\lambda m\eins_{\{X>0\}}\big).
\]
Therefore, by assumption,
\[
\DR\big(X+m\eins_{\{X>0\}}\big)\leq \inf_{\lambda\in (0,1)}\DR\big(X+\lambda m\eins_{\{X>0\}}\big)=\DR(X).
\]
In order to prove the implication (ii)$\,\Rightarrow\,$(i), let $X\in \Bb$ with $X\geq 0$ and $\lambda>0$. If $X=0$ or $\lambda =1$, it follows that $\DR(\lambda X)=\DR(X)$. Therefore, assume that $\sup X>0$ and $\lambda \neq 1$. First, we consider the case, where $\lambda> 1$.\ Then, $\lambda X\geq X$, so that $\DR(\lambda X)\geq \DR(X)$. Moreover, $\lambda X\leq X+(\lambda-1)\sup X\eins_{\{X>0\}}$. Defining $m:=(\lambda-1)\sup X$, it follows that
\[
\DR(\lambda X)\leq \DR\big(X+m\eins_{\{X>0\}}\big)=\DR(X).
\]
Now, let $\lambda<1$. Then, $\DR(\lambda X)\leq \DR(X)$. On the other hand,
$$X\leq \lambda X+(1-\lambda )(\sup X)\eins_{\{X>0\}}=\lambda X+(1-\lambda )(\sup X)\eins_{\{\lambda X>0\}}.$$
Therefore, defining $m:=(1-\lambda)\sup X$, we find that
\[
\DR(X)\leq \DR\big(\lambda X+m\eins_{\{\lambda X>0\}}\big)=\DR(\lambda X).
\]
By Proposition \ref{prop.condition}, (i) implies (iii) and, trivially, (iii) implies (i). 
\end{proof}
\begin{proof}[Proof of Lemma \ref{lem.X_plus}]
Trivially, (i)$\,\Rightarrow\,$(ii) and (ii)$\,\Rightarrow\,$(iii). In order to prove the remaining implication (iii)$\,\Rightarrow\,$(i), let $X\in \Bb$. Then, $X\leq X^+$ and therefore $\DR(X)\leq \DR(X^+)$. If $X\geq 0$, then $X=X^+$ and it follows that $\DR(X)=\DR(X^+)$. Therefore, assume that $\inf X<0$ and define $m:=-\inf X$. Then,
\[
X^+-m \eins_{\{X^+= 0\}}=X^+-m \eins_{\{X\leq 0\}}\leq X.
\]
By assumption, we obtain that
\[
\DR(X^+)=\DR\big(X^+-m \eins_{\{X^+= 0\}}\big)\leq \DR(X).
\]
\end{proof}
\section{Proofs of Section \ref{sc:Choquet}}
\label{proofs of sc.5}
\begin{proof}[Proof of Proposition  \ref{prop. sc 5}]
 \JS{The implication (i)$\, \Rightarrow\, $(ii) follows from Remark \ref{rem.choquet} and the implication (ii)$\, \Rightarrow\, $(i) is trivial.}
\end{proof}
\begin{proof}[Proof of Proposition \ref{prop.contbelow}]
Clearly, (iii) implies (ii).\ The implication (ii)$\,\Rightarrow\,$(i) follows directly from the fact that, for any sequence $(X_n)_{n\in \N}\subset \Bb$ with $X_n\nearrow X\in \Bb$ as $n\to \infty$,  $\left\{X_n>0\right\}\subset \left\{X_{n+1}>0\right\}$ for all $n\in \N$ and $$\bigcup_{n\in \N}\left\{X_n>0\right\}=\{ X>0\}.$$
 It remains to prove the implication (i)$\,\Rightarrow\,$(iii). Let $(X_n)_{n\in \N}\subset \Bb$ and $X\in \Bb$ with $X_n \nearrow X$ as $n\to \infty$. \JS{By potentially adding $\| X_1\|_\infty$ to $X$ and $X_n$ for all $n\in \N$ and using the fact that the Choquet integral is a monetary risk measure, we may w.l.o.g.\ assume that $X_1\geq 0$.} Then, using the monotone convergence theorem,
\[
\lim_{n\to \infty}\int X_n\, \d c=\lim_{n\to \infty}\int_0^\infty\DR(X_n-s)\, \d s=\int_0^\infty\DR(X-s)\, \d s=\int X\, \d c.
\]
The proof is complete.
\end{proof}
\begin{proof}[Proof of Proposition \ref{prop. polish}]
Clearly (ii) implies (i).\ We prove the nontrivial implication (i)$\,\Rightarrow\,$(ii).\ Let $\mathcal O$ denote the set of all open subsets of $\Omega$ and $c(B):=\DR(\eins_B)$ for all $B\in \mathcal O$. Since $\DR$ is submodular with \eqref{eq.replb}, it follows that
$$
\DR(X)=c\big(\{X>0\}\big)\quad\text{for all }X\in \Lb.
$$
The continuity from below of $\DR$ implies that $
c\big(\bigcup_{n\in \N}B_n\big)= \lim_{n\to \infty}c(B_n)$ for all sequences $(B_n)_{n\in \N}\subset \mathcal O$ with $B_n\subset B_{n+1}$ for all $n\in \N$. The statement now follows from \cite[Corollary 2.6]{nendel}.
\end{proof}
\begin{proof}[Proof of Proposition \ref{prop.contabove}]
We first prove the implication (i)$\,\Rightarrow\,$(iii). Let $(X_n)_{n\in \N}\subset \Bb$ with $X_n \searrow 0$ as $n\to \infty$. Using the monotone convergence theorem,
\[
\lim_{n\to \infty}\int X_n\, \d c=\lim_{n\to \infty}\int_0^\infty\DR(X_n-s)\, \d s=0.
\]
Clearly, (iii) implies (ii), and it remains to prove that (ii) implies (i).\ To that end, observe that, for every sequence $(X_n)_{n\in \N}\subset \Bb$ with $X_n \searrow 0$ as $n\to \infty$ and $\varepsilon >0$, $\{X_{n+1}>\varepsilon\}\subset \{X_n>\varepsilon\}$ for all $n\in \N$ and $\bigcap_{n\in \N}\{X_n>\varepsilon\}=\emptyset$. Hence,
\[
\lim_{n\to \infty} \DR\left( X_n- \varepsilon \right)=\lim_{n\to \infty}\DR\big( \eins_{\{X_n> \varepsilon\}} \big)=0.
\]
The representation via (countably additive) probability measures now follows from the standard theory on coherent risk measures, cf.\ \cite{MR3859905}.
\end{proof}
\section{Proofs of Section \ref{sec:distorted}}\label{proofs of sc.6}

\begin{proof}[Proof of Theorem \ref{Theorem: law inv, Dist}]
Clearly, (ii) implies (i). For the other implication, let 
\[
P:=\big\{p\in [0,1]\, \big|\, \exists X\in C\colon \P(X>0)=p\big\} 
\]
Thus, for any $p\in P$, there exists some $X_p\in C$ with $\P\left(X_p>0 \right)=p$ and, by our global assumption on $C$, $I_p:=\eins_{\{X_p>0\}}\in C$. For $p\in [0,1]$, we define $$T(p):=\DR(I_{q_p})\quad\text{with}\quad q_p:=\sup \big([0,p]\cap P\big).$$ As a result $T(p)=\DR(I_p)$ for all $p\in P$. In particular, $T(0)=0$ and $T(1)=1$.\ Let $X\in C$ and $p_X:=\PD_\P(X)=\P(X>0)$. Then, by assumption, 
\[
\DR(X)=\DR\big(\eins_{\{X>0\}}\big)=\DR(I_{p_X})=T(p_X )=T\big(\PD_\P(X)\big).
\]

\end{proof}

\begin{proof}[Proof of Theorem \ref{thm. eq dist. PD}]
In view of Theorem \ref{Theorem: law inv, Dist}, we only have to prove that the distortion function in (ii) is nondecreasing.\ To that end, assume that $\DR(X)=T\big(\PD_\P(X)\big)$ for all $X\in C$ with a distortion function $T\colon [0,1]\to [0,1]$.\ Let $p,q\in P$.\ Since $C$ contains an ordered subset, there exist $X,Y\in C$ with $\PD_\P(X)=p$, $\PD_\P(Y)=q$, and $\{X>0\}\subset \{Y>0\}$. Due to the monotonicity of $\DR$, it follows that $$T(p)=\DR(X)=\DR\big(\eins_{\{X>0\}}\big)\leq \DR\big(\eins_{\{Y>0\}}\big)=\DR(Y)=T(q).$$
The proof is complete.
\end{proof}
\begin{proof}[Proof of Proposition \ref{prop. atomless}]
\JS{First assume that $\Q$ satisfies \eqref{eq.cond.distineq}.} Then,
\[
T(p)\geq \int_0^pq_\Q(1-s)\, \d s= \int_{1-p}^1q_\Q(s)\, \d s\quad\text{for all }p\in [0,1].
\]
Therefore, using \cite[Lemma 4.60]{MR3859905},
\[
T\big(\PD_\P(X)\big)\geq \int_{1-\PD_\P(X)}^1q_\Q(s)\, \d s\geq \PD_\Q(X)\quad\text{for all }X\in \Bb.
\]
On the other hand, if $(\Omega,\FF,\P)$ is atomless and $\Q$ satisfies \eqref{eq.distineq}, it follows from \eqref{eq.rep.distorted} that
\[
T(p)\geq \int_{1-p}^1q_\Q(s)\, \d s=\int_0^pq_\Q(1-s)\, \d s\quad\text{for all }p\in [0,1].
\]
\end{proof}

\end{document}